\documentclass[letter,11pt]{article}
\usepackage{amsmath,amssymb,amsthm,color,graphicx}
\usepackage{euscript}
\usepackage{times}

\newtheorem{theorem}{Theorem}[section]

\newtheorem{proposition}[theorem]{Proposition}
\newtheorem{lemma}[theorem]{Lemma}

\graphicspath{{Figs/}}

\def\A{\mathcal{A}}
\def \reals{{\mathbb R}}
\def \L{{L}}

\def\eps{{\varepsilon}}
\def\eps{{\varepsilon}}
\def\poly{\diamond}

\newcommand{\ignore}[1]{}

\def\X{\mathcal{X}}
\def\Y{\mathcal{Y}}

\def\DT{\mathop{\mathrm{DT}}}
\def\VD{\mathop{\mathrm{VD}}}
\def\Vor{\mathop{\mathrm{Vor}}}

\begin{document}

\title{On topological changes in the Delaunay triangulation of moving points}
\author{Natan Rubin\thanks{%
Department of Mathematics and Computer Science, Freie Universit\"{a}t Berlin, Takustr. 9, 14195, Berlin, Germany.
Email: {\tt rubinnat@tau.ac.il}. Work on this paper was partly supported by Minerva Postdoctoral Fellowship from Max Planck Society, and by Grant 338/09 from the Israel Science Fund.}}

\maketitle
\begin{abstract}
Let $P$ be a collection of $n$ points moving along pseudo-algebraic trajectories in the plane.\footnote{So, in particular, threre are constants $s,c>0$ such that any four points are co-circular at most $s$ times, and any three points are collinear at most $c$ times.} 
One of the hardest open problems in combinatorial and computational geometry is to obtain a nearly quadratic upper bound, or at least a subcubic bound, on the maximum number of discrete changes that the Delaunay triangulation $\DT(P)$ of $P$ experiences during the motion of the points of $P$.

In this paper we obtain an upper bound of $O(n^{2+\eps})$, for any $\eps>0$, under the assumptions that (i) any four points can be co-circular at most twice, and (ii) either no triple of points can be collinear more than twice, or no ordered triple of points can be collinear more than once.
\end{abstract}

\section{Introduction}\label{Sec:Intro}
\paragraph{Delaunay triangulations.} 
Let $P$ be a finite set of points in the plane. 
Let $\Vor(P)$ and $\DT(P)$ denote the Voronoi diagram and Delaunay
triangulation of $P$, respectively. For a point $p \in P$, let
$\Vor(p)$ denote the Voronoi cell of $p$. 
The Delaunay triangulation $\DT=\DT(P)$ consists of all 
triangles spanned by $P$ whose circumcircles do not contain points of $P$ in their
interior. A pair of points $p,q\in P$ is connected by a Delaunay edge if and only if
there is a circle passing through $p$
and $q$ that does not contain any point of $P$ in its interior.
%
Delaunay triangulations and  their duals, Voronoi diagrams, are fundamental to much 
of computational geometry and its applications. 
See \cite{AK,Ed2} for a survey and a
textbook on these structures.

In many applications of Delaunay/Voronoi methods (e.g., mesh generation and kinetic collision detection) the points of the input set $P$ are moving continuously, so
these diagrams need to be efficiently updated during the motion.
Even though the motion of the points is continuous, the combinatorial and topological structure of the Voronoi and
Delaunay diagrams change only at
discrete times when certain critical events occur. 

For the purpose of kinetic maintenance, Delaunay triangulations are 
nice structures, because, as mentioned above, they admit local 
certifications associated with individual triangles.  This makes 
it simple to maintain $\DT(P)$ under point motion: an update is 
necessary only when one of these empty circumcircle conditions 
fails---this corresponds to co-circularities of certain subsets of
four points.\footnote{We assume that the motion of the points is sufficiently generic, so that no more than four points can become co-circular at any given time.} Whenever such an event happens, 
a single edge flip easily restores Delaunayhood.

Let $n$ be the number of moving points in $P$. 
We assume that the
points move with so-called pseudo-algebraic motions of constant description complexity, meaning (in particular) that any four points are co-circular at most $s$ times, for some constant $s>0$.  
By using lower-envelope techniques in certain parametric planes, 
Fu and Lee \cite{FuLee} and
Guibas et al.~\cite{gmr-vdmpp-92} show roughly cubic upper bounds on the number of discrete (also known as \textit{topological}) changes in $\DT(P)$. 
The latter study~\cite{gmr-vdmpp-92} obtains an upper bound of
$O(n^2 \lambda_{s+2}(n))$, where $\lambda_s(n)$ is the maximum length
of an $(n,s)$-Davenport-Schinzel sequence~\cite{SA95}.

If each point of $P$ is moving along a straight line and with the same speed, a slightly better upper bound of $O(n^3)$ can be established for the number of discrete changes experienced by $\DT(P)$ (see, e.g., \cite{Vladlen}). A substantial gap exists between these upper bounds
and the best known quadratic lower bound~\cite{SA95}. Closing this gap has been in the 
computational geometry lore for many years, and is considered as one of the major (and very difficult) problems in the field; see \cite{TOPP}.


The instances of the general problem for which the number of discrete changes in $\DT(P)$ is provably sub-cubic, are strikingly few. It is worth mentioning the result of Koltun \cite{Vladlen} which deals with sets of points moving along straight lines with equal speeds such that all points start their motion from a fixed line. In this particular case one can show that any four points are co-circular at most twice, and any three points are co-linear at most once. (The result of \cite{Vladlen} is not topological and relies on the equations of point trajectories.)

Due to the very slow progress on the above general problem, several alternative lines of study have emerged in the last two decades.

Chew \cite{Chew} proved that the Voronoi diagram undergoes only a near-quadratic number of discrete changes if it is defined with respect to a so called ``polygonal" distance function. The dual representation of such a diagram $\VD^\poly(P)$ yields a proper triangulation of a certain connected subregion of the convex hull of $P$. Agarwal et al.~\cite{Stable} use the above polygonal structures
to efficiently maintain the so called {\it $\alpha$-stable} subgraph $\DT(P)$, whose edges are robust with respect to small changes in the underlying norm. 

Another line of research \cite{ABGHZ,AWY,KRS} asks if one can define (and efficiently maintain) a proper triangulation of the convex hull of $P$, which would change only near-quadratically many times during the motion of $P$. 
The most recent such study \cite{KRS} provides a (relatively) simple such triangulation which undergoes, in expectation, only $O(n^2\lambda_{c+2}(n)\log^2n)$ discrete changes (where $c$ is the maximum possible number of collinearities defined by any three points of $P$).


\smallskip
\noindent{\bf Our result.}
We study the case in which (i) any four points of $P$ are co-circular at most {\it twice} during the motion, and (ii) either every unordered triple can be collinear at most twice or every ordered triple of points can be collinear at most once\footnote{That is, there can be only one collinearity of an ordered triple $(p,q,r)$ so that the points appear in this order along the common line.}, and derive a nearly tight upper bound of $O(n^{2+\eps})$, for any $\eps>0$, on the number of discrete changes experienced by $\DT(P)$ during the motion in either of these cases.
We believe that our results constitute a substantial progress towards establishing nearly quadratic (or just sub-cubic) bounds for more general instances of the problem, such as the simple and natural instance of points moving along straight lines with equal speeds. In this case any four points admit at most {\it three} co-circularities, and any triple of points can be collinear at most twice. We believe that the tools developed in this paper can be extended to tackle this instance, and possibly also the general case.

\smallskip
\noindent{\bf Proof overview and organization.}
The majority of the discrete changes in $\DT(P)$ occur at moments $t_0$ when some four points $p,q,a,b\in P$ are co-circular, and the corresponding circumdisc contains no other points of $P$. We refer to these events as Delaunay co-circularities. Suppose that $p,a,q,b$ appear along their common circumcircle in this order, so $ab$ and $pq$ form the chords of the quadrilateral spanned by these points. Right before $t_0$, one of the chords, say $pq$, is Delaunay and thus admits a $P$-empty disc whose boundary contains $p$ and $q$. 
Right after time $t_0$, the edge $pq$ is replaced in $\DT(P)$ by $ab$.  
Informally, this happens because the Delaunayhood of $pq$ is violated by $a$ and $b$: Any disc whose boundary contains $p$ and $q$ contains at least one of the points $a,b$. 
If $pq$ does not re-enter $\DT(P)$ after time $t_0$, we can charge the event at time $t_0$ to the edge $pq$. We thus assume that $pq$ is again Delaunay at some moment $t_1>t_0$. In particular, at that moment the Delaunayhood of $pq$ is no longer violated by $a$ and $b$. Before this happens, either at least one of $a$ or $b$ must hit $pq$, or an additional co-circularity of $a,b,p,q$ must occur during $(t_0,t_1)$. Using our assumption that $a,b,p,q$ induce at most two co-circularity events, we can guarantee (up to a reversal of the time axis) that the co-circularity at time $t_0$ is the last co-circularity of these four points. Thus, one of $a,b$, let it be $a$, must cross $pq$ during $(t_0,t_1)$.

Our goal is to derive a recurrence formula for the maximum number $N(n)$ of such Delaunay co-circularities induced by any set $P$ of $n$ points (whose motion satisfies the above conditions). 

As a preparation, we study, in Section \ref{Sec:Prelim}, the set of all co-circularities that involve the disappearing Delaunay edge $pq$ and some other pair of points of $P\setminus \{p,q\}$ and occur during the period $(t_0,t_1)$ when $pq$ is absent from $\DT(P)$.
This is done in a fairly general setting, where any four points of $P$ can be co-circular, and any three points 
of $P$ can be collinear, at most constantly many times. 
Along the way, we establish several structural results which (as we believe) are of independent interest.

In Section \ref{Sec:DelCocircs} we use the general machinery of Section \ref{Sec:Prelim} to obtain a recurrence formula for $N(n)$ in the case where any four points of $P$ are co-circular at most twice.
Recall that $pq$ leaves $\DT(P)$ at such a Delaunay co-circularity, at some time $t_0$, in order to re-enter $\DT(P)$ at some later time $t_1>t_0$, and $pq$ is hit by the point $a$ in the interval $(t_0,t_1)$ of its non-Delaunayhood.

If we find at least $\Omega(k^2)$ ``shallow" co-circularities\footnote{Each of these co-circularities would become a Delaunay co-circularity after removal of at most $k$ points of $P$.}, whose respective circumdiscs (i) touch $p$ and $q$, and (ii) contain at most $k$ points of $P$, we charge them for the disappearance of $pq$. We use the standard probabilistic technique of Clarkson and Shor \cite{CS} to show that the number of Delaunay co-circularities, for which our simple charging works, is $O\left(k^2N(n/k)\right)$.
Informally, such Delaunay co-circularities contribute a nearly quadratic term to the overall recurrence formula (see, e.g., \cite{ASS} and \cite{ConstantLines}). Similarly, if we find a ``shallow" collinearity of $p,q$ and another point (one halfplane bounded by the line of collinearity contains at most $k$ points) we charge the disappearance of $pq$ to this collinearity. A combination of the Clarkson-Shor technique with the known near-quadratic bound on the number of topological changes in the convex hull of $P$ (see \cite[Section 8.6.1]{SA95}) yields a near-quadratic bound in this case.

It thus remains to bound the number of Delaunay co-circularities for which $p$ and $q$ participate in fewer ``shallow" co-circularities and in no ``shallow" collinearity during $(t_0,t_1)$. In this case, using the general properties established in Section \ref{Sec:Prelim}, one can restore the Delaunayhood of $pq$ throughout $(t_0,t_1)$ by removal of some subset $A$ of $O(k)$ points of $P$. In particular, the point $a$, which crosses $pq$, must belong to $A$.
In the smaller Delaunay triangulation $\DT\left((P\setminus A) \cup\{a\}\right)$, the edge $pq$ undergoes a complex process referred to as a {\it Delaunay crossing by} $a$. 

In Section \ref{Sec:CrossOnce}, we derive a recurrence formula for the number of these Delaunay crossings. 
This is achieved by establishing several structural properties of these novel configurations.
Combined with the analysis of Section \ref{Sec:Prelim}, this yields the desired ``near-quadratic" recurrence for the number of Delaunay co-circularities. 


\section{Preliminaries}\label{Sec:Prelim}
In this section we define the basic notions regarding Delaunay triangulations of moving points, and introduce some of the key techniques which will be repeatedly used in the rest of the paper.

\paragraph{Delaunay co-circularities.} Let $P$ be a collection of $n$ points moving along pseudo-algebraic trajectories in the plane. That is, there exist constants $s$ and $c$ so that any four points are co-circular at most $s$ times, and any three points are collinear at most $c$ times. (As far as this section is concerned, we do not impose any further restrictions on the choice of $s$ and $c$, except for their being constant.)

We may assume, without loss of generality, that the trajectories of the points of $P$ satisfy all the standard general position assumptions. That is, no five points can become co-circular during the motion, no four points can become collinear, no two points can coincide, and no two events of either a co-circularity of four points or of collinearity of three points can occur simultaneously.
In addition, we assume that in every co-circularity event involving some four points $a,b,p,q\in P$, each of the points, say $a$, crosses the circumcircle of the other three points $b,p,q$; that is, it lies outside the circle right before the event and inside right afterwards, or vice versa. Similarly, we assume that in every collinearity event involving some triple of points of $P$, each of the points crosses the line through the remaining two points.
Degeneracies in the point trajectories of the above kinds can be handled, both algorithmically and combinatorially, by any of the standard symbolic perturbation techniques, such as simulation of simplicity \cite{EM}; for combinatorial purposes, a sufficiently small generic perturbation of the motions will get rid of any such degeneracy, without decreasing the number of topological changes in the diagram.

\begin{figure}[htbp]
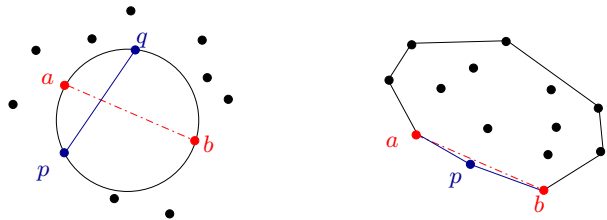

\begin{center}
\input{DelaunayCocirc.pstex_t}\hspace{2cm}\input{HullEvent.pstex_t}
\caption{\small{Left: A Delaunay co-circularity of $a,b,p,q$. An old Delaunay edge $pq$ is replaced by the new edge $ab$. Right: A collinearity of $a,p,b$ right before $p$ ceases being a vertex on the boundary of the convex hull. 
}} \label{Fig:DelaunayEvents}
\end{center}
\end{figure} 

The Delaunay triangulation $\DT(P)$ changes at discrete time moments $t_0$ when one of the following two types of events occurs.

(i) Some four points $a,b,p,q$  of $P$ become co-circular, so that the cicrumdisc of $p,q,a,b$ is {\it empty}, i.e., does not contain any point of $P$ in its interior. We refer to such events as {\it Delaunay co-circularities}, to distinguish them from non-Delaunay co-circularities, for which the circumdisc of $a,b,p,q$ is nonempty, that is, contains one or more points of $P$ in its interior.\footnote{Strictly speaking, $\DT(P)$ is not a triangulation at the time $t_0$ of such a co-circularity, because it contains then a pair of {\it crossing} edges, say $ab$ and $pq$ (as depicted in Figure \ref{Fig:DelaunayEvents} (left)).} See Figure \ref{Fig:DelaunayEvents} (left).

In what follows, we shall use $N(n)$ to denote the maximum possible number of Delaunay co-circularities induced by {\it any} set $P$ of $n$ points whose motion satisfies the above general assumptions.\footnote{In the subsequent sections, we shall impose additional restrictions on the pseudo-algebraic motions of the points of $P$, thereby  redefining $N(n)$.}

(ii) Some three points $a,b,p$ of $P$ become collinear on the boundary of the convex hull of $P$. Assume that $p$ lies between $a$ and $b$. In this case, if $p$ moves into the interior of the hull, then, right after this event, the triangle $abp$ becomes a new Delaunay triangle. Similarly, if $p$ moves outside and becomes a new vertex, then, right before this event, the old Delaunay triangulation $\DT(P)$ contained the old Delaunay triangle $abp$, which has shrinked to a segment and disappeared at the event. See Figure \ref{Fig:DelaunayEvents} (right).
The number of such collinearities on the convex hull boundary is known to be at most nearly quadratic; see, e.g., \cite[Section 8.6.1]{SA95} and below.


\paragraph{Shallow co-circularities and the Clarkson-Shor argument.}
We say that a co-circularity event has {\it level} $k$ if its corresponding circumdisc contains exactly $k$ points of $P$ in its interior. In particular, the Delaunay co-circularities have level $0$. The co-circularities having level at most $k$ are called {\it $k$-shallow}.

We can express the maximum possible number of $k$-shallow co-circularities in $P$ in terms of the more elementary quanitity $N(n/k)$ via the following fairly general argument, first introduced by Clarkson and Shor \cite{CS}.
(With no loss of generalty, we assume that $k\geq 1$, for otherwise we can trivially bound the maximum number of Delaunay, that is, $0$-shallow co-circularities in $P$ by $N(n)$.)

Let $t_0$ be the time of a $k$-shallow co-circularity which involves some four points $p,q,a,b$ in $P$, and let $A_0$ denote the set of at most $k$ points that lie at time $t_0$ in the interior of the common circumdisc of $p,q,a,b$. Note that the above co-circularity is Delaunay with respect to $P\setminus A_0$, and with respect to any subset $R$ of $P\setminus A_0$ which contains $p,q,a,b$.

We sample at random (and without replacement) a subset $R\subset P$ of $O(n/k)$ points. As is easy to check, the following two events occur {\it simultaneously} with probability at least $\Theta(1/k^4)$: (1) the sample $R$ contains the four points $p,q,a,b$, and (2) none of the points of $A_0$ belongs to $R$. (An explicit calculation of the above probability can be found in several classical texts, such as \cite{CS} or \cite{SA95}.)

In the case of success, the aforementioned $k$-shallow co-circularity in $P$ becomes a Delaunay co-circularity with respect to $R$. 
Hence, the overall number of $k$-shallow co-circularities in $P$ is $O(k^4 N(n/k))$.

\paragraph{Shallow collinearities.} Similar notations apply to collinearities of triples of points $p,q,r$. A collinearity of $p,q,r$ is called {\it $k$-shallow} if the number of points of $P$ to the left, or to the right, of the line through $p,q,r$ is at most $k$. 

The (essentially) same probabilistic argument implies that the number of such events, for $k\geq 1$, is $O(k^3L(n/k))$, where $L(m)$ denote the maximum number of discrete changes on the convex hull of an $m$-point subset of $P$. (The difference in the exponent of $k$ follows because now each configuration at hand involves only three points.)

As shown, e.g., in \cite[Section 8.6.1]{SA95},  $L(m)=O(m^2\beta(m))$, where $\beta(\cdot)$ is an extremely slowly growing function.\footnote{Specifically,
$\beta(n)=\frac{\lambda_{s+2}(n)}{n}$, where $\lambda_{s+2}(n)$ is the maximum length of an $(n,s+2)$-Davenport-Schinzel sequence (see Section \ref{Sec:Intro}), and $s$ is the maximum number of collinearities of any fixed triple of points. The pseudo-algebraicity of the motion implies that $s$ is a constant, but we will restrict $s$ further; see below.} We thus get that the number of $k$-shallow collinearities is $O(kn^2\beta(n/k))=O(kn^2\beta(n))$. 

\begin{figure}[htbp]
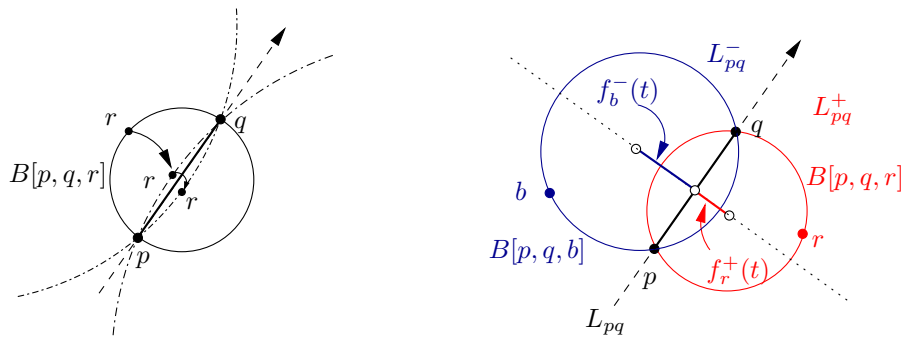

\begin{center}
\input{Circumdisc.pstex_t}\hspace{2cm}\input{RedBlueFuncs.pstex_t}
\caption{\small Left: The circumdisc $B[p,q,r]$ of $p,q$ and $r$ moves continuously as long as these three points are not collinear, and then flips over to the other side of the line of collinearity after the collinearity. 
Right: A snapshot at moment $t$. 
In the depicted configuration we have $f_b^-(t)<0<f_r^+(t)$.}
\label{Fig:RedBlueFuncs}
\vspace{-0.5cm}
\end{center}
\end{figure}

\paragraph{The red-blue arrangement.} For every pair of points $p,q$ of $P$ we construct a two-dimensional arrangement which ``encodes" all the collinearities and co-circularities that involve $p$ and $q$ (together with one or two additional points). This is done as follows.

For every ordered pair $(p,q)$ of points of $P$, we
denote by $\L_{pq}$ the line passing through $p$ and $q$ and oriented from $p$ to $q$. 
Define $\L_{pq}^-$ (resp., $\L_{pq}^+$) to be the halfplane to the left (resp., right) of $\L_{pq}$.
Notice that $\L_{pq}$ moves continuously with $p$ and $q$ (since, by assumption, $p$ and $q$ never coincide during the motion). Note also that $\L_{pq}$ and $\L_{qp}$ are oppositely oriented and that $\L_{pq}^+=\L_{qp}^-$ and $\L_{pq}^-=\L_{qp}^+$.
Accordingly, we orient the edge $pq$ connecting $p$ and $q$ from $p$ to $q$, so that the edges $pq$ and $qp$ have opposite orientations.

Any three points $p,q,r$ span a circumdisc $B[p,q,r]$ which moves continuously with $p,q,r$ as long as $p,q,r$ are not collinear. See Figure \ref{Fig:RedBlueFuncs} (left). When $p,q,r$ become collinear, say, when $r$ crosses $pq$ from $\L_{pq}^-$ to $\L_{pq}^+$, the circumdisc $B[p,q,r]$ changes instantly from being all of $\L_{pq}^+$ to all of $\L_{pq}^-$.
Similarly, when $r$ crosses $\L_{pq}$ from $\L_{pq}^-$ to $\L_{pq}^+$ {\it outside} $pq$, the circumdisc changes instantly from $\L_{pq}^-$ to $\L_{pq}^+$. Symmetric changes occur when $r$ crosses $\L_{pq}$ from $\L_{pq}^+$ to $\L_{pq}^-$.

For a fixed ordered pair $p,q\in P$, we call a point $a$ of $P\setminus\{p,q\}$ {\it red} (with respect to the oriented edge $pq$) if $a\in \L_{pq}^+$; otherwise it is {\it blue}.

As in \cite{gmr-vdmpp-92}, we define, for each $r\in P\setminus\{p,q\}$, a pair of partial functions $f_r^+,f_r^-$ over the time axis as follows.
If $r\in \L_{pq}^+$ at time $t$ then $f_r^-(t)$ is undefined, and $f^+_r(t)$ is the signed distance of the center $c$ of $B[p,q,r]$ from $\L_{pq}$; it is positive (resp., negative) if $c$ lies in $\L_{pq}^+$ (resp., in $\L_{pq}^-$). A symmetric definition applies when $r\in \L_{pq}^-$. Here too $f^-_r(t)$ is positive (resp., negative) if the center of $B[p,q,r]$ lies in $\L_{pq}^+$ (resp., in $\L_{pq}^-$). We refer to $f_r^+$ as the {\it red function} of $r$ (with respect to $pq$) and to $f_r^-$ as the {\it blue function} of $r$. Note that at all times when $p,q,r$ are not collinear, exactly one of $f_r^+,f_r^-$ is defined. See Figure \ref{Fig:RedBlueFuncs} (right).
The common points of discontinuity of $f_r^+,f_r^-$ occur at moments when $r$ crosses $\L_{pq}$. Specifically, $f_r^+$ tends to $+\infty$ before $r$ crosses $\L_{pq}$ from $\L_{pq}^+$ to $\L_{pq}^-$ outside the segment $pq$, and it tends to $-\infty$ when $r$ does so within $pq$; the behavior of $f_r^-$ is fully symmetric.

Let $E^+$ denote the lower envelope of the red functions, and let $E^-$ denote the upper envelope of the blue functions. The edge $pq$ is a Delaunay edge at time $t$ if and only if $E^-(t)<E^+(t)$. Any disc whose bounding circle passes through $p$ and $q$ which is centered anywhere in the interval $(E^-(t),E^+(t))$ along the perpendicular bisector of $pq$ is empty at time $t$, and thus serves as a witness to $pq$ being Delaunay.
If $pq$ is not Delaunay at time $t$, there is a pair of a red function $f_r^+(t)$ and a blue function $f_b^-(t)$ such that $f_r^+(t)<f_b^-(t)$.
For example, we can take $f_r^+$ (resp., $f_b^-$) to be the function attaining $E^+$ (resp., $E^-$) at time $t$. In such a case, we say that the Delaunayhood of $pq$ is {\it violated} by the pair of points $r,b\in P$ which define $f_r^+,f_b^-$. See Figure \ref{Fig:Envelopes}. Note that in general there can be many pairs $(r,b)$ that violate $pq$ (quadratically many in the worst case).

\begin{figure}[htbp]
\begin{center}
\input{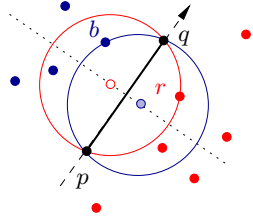}
\caption{\small Snapshot at a fixed moment $t$: The red envelope $E^+$ coincides with the red function $f_r^+$. The blue envelope coincides with the blue function $f^-_b(t)$. Note that $pq$ is not a Delaunay edge because $E^+(t)$ (represented by the hollow center) is smaller than $E^-(t)$ (represented by the shaded center).}
\label{Fig:Envelopes}
\end{center}
\end{figure}

Hence, at any time when the edge $pq$ joins or leaves $\DT(P)$, via a Delaunay co-circularity involving $p$, $q$, and two other points of $P$, we have $E^-(t)=E^+(t)$. In this case the two other points, $a,b$, are such that one of them, say $a$, lies in $\L_{pq}^+$ and $b$ lies in $\L_{pq}^-$, and $E^+(t)=f_a^+(t), E^-(t)=f_b^-(t)$.

\smallskip
\noindent{\it Remark.} 
Right before the edge $pq$ is crossed by a red point $r$, the corresponding function $f_r^+$ lies below all the blue functions $f_b^-$ (if they exist), so the Delaunayhood of $pq$ is violated by each of the subsequent pairs $(r,b)$. 
In other words, the edge $pq$ cannot be Delaunay right before (resp., after) being hit by a point of $P$, unless $pq$ joins or leaves the convex hull of $P$. 

\smallskip

Let $\A=\A_{pq}$ denote the arrangement of the $2n-4$ functions $f_r^+(t),f_r^-(t)$, for $r\in P\setminus\{p,q\}$, drawn in the parametric $(t,\rho)$-plane, where $t$ is the time and $\rho$ measures signed distance along the perpendicular bisector of $pq$. We label each vertex of $\A$ as red-red, blue-blue, or red-blue, according to the colors of the two functions meeting at the vertex. Note that our general position assumptions imply that $\A$ is also in general position, so that no three function pass through a common vertex, and no pair of functions are tangent to each other. Note also that the functions forming $\A$ have in general discontinuities, at the corresponding collinearities. At each such collinearity, a red function $f_r^+$ tends to $\infty$ or $-\infty$ on one side of the critical time, and is replaced on the other side by the corresponding blue function $f_r^-$ which tends to $-\infty$ or $\infty$, respectively.

An intersection between two red functions $f_a^+,f_b^+$ corresponds to a co-circularity event which involves $p,q,a$ and $b$, occurring when both $a$ and $b$ lie in $\L_{pq}^+$.
Similarly, an intersection of two blue functions $f_a^-,f_b^-$ corresponds to a co-circularity event
involving $p,q,a,b$ where both $a$ and $b$ lie in $\L_{pq}^-$. Also, an intersection of a red fuction $f_a^+$ and a blue function $f_b^-$ represents a co-circularity of $p,q,a,b$, where $a\in \L_{pq}^+$ and $b\in \L_{pq}^-$. We label these co-circularities, as we labeled the vertices of $\A$, as red-red, blue-blue, and red-blue (all with respect to $pq$), depending on the respective colors of $a$ and $b$.

It is instructive to note that in any co-circularity of four points of $P$ there are exactly two pairs (the opposite pairs in the co-circularity)
with respect to which the co-circularity is red-blue, and four pairs (the adjacent pairs) with respect to which the co-circularity is ``monochromatic". Suppose that the above co-circularity is Delaunay. Then the two pairs for which the co-circularity is red-blue are those that enter or leave the Delaunay triangulation $\DT(P)$ (one pair enters and one leaves). The Delaunayhood of pairs for which the co-circularity is monochromatic is not affected by the co-circularity, which appears in the corresponding arrangement as a {\it breakpoint} of either $E^+(t)$ or of $E^-(t)$.

\begin{figure}[htbp]
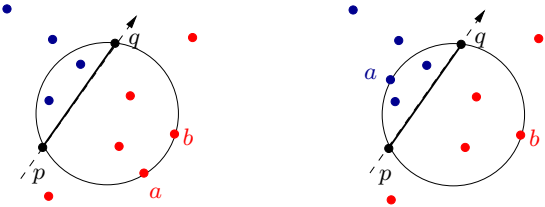

\begin{center}
\input{RedRedCocirc.pstex_t}\hspace{2cm}\input{RedBlueCocirc.pstex_t}
\caption{\small Intersections between two red functions $f_a^+$ and $f_b^+$ (left), or a blue function $f_a^-$ and a red function $f_b^+$ (right), correspond to red-red or red-blue co-circularities.}
\label{Fig:RedBlueCocirc}
\end{center}
\vspace{-0.4cm}
\end{figure}

The main weakness of the previous approaches \cite{FuLee,gmr-vdmpp-92} is that they study only the lower envelope $E^+(t)$ of red functions, and the upper envelope $E^-(t)$ of blue functions, which are merely substructures within the above arrangement $\A_{pq}$.
This yields a roughly linear upper bound on the number of monochromatic co-circularities with respect to the edge $pq$ under consideration.
Repeating the same argument for the $n(n-1)$ possible (oriented) edges $pq$ then results in a far too high, near-cubic upper bound on the number of Delaunay co-circularities.

Instead, we exploit the underlying structure of $\A_{pq}$ in order to establish the following
main technical result of this section.

\begin{theorem}\label{Thm:RedBlue}
Let $P$ be a collection of $n$ points moving as described above. Suppose that an edge $pq$ belongs to $\DT(P)$ at (at least) one of two moments $t_0$ and $t_1$, for $t_0<t_1$.
Let $k>12$ be some sufficiently large constant.\footnote{The constants in the $O(\cdot)$ and $\Omega(\cdot)$ notations do not depend on $k$.}
Then one of the following conditions holds:\\
\indent (i) There is a $k$-shallow collinearity which takes place during $(t_0,t_1)$, and involves $p$, $q$ and another point $r$.\\
\indent (ii) There are $\Omega(k^2)$ $k$-shallow red-red, red-blue, or blue-blue co-circularities (with respect to $pq$) which occur during $(t_0,t_1)$.\\
\indent (iii) There is a subset $A\subset P$ of fewer than $3k$ points whose removal guarantees that $pq$ belongs to $\DT(P\setminus A)$ throughout $(t_0,t_1)$.
\end{theorem}

Notice that we do not assume in Theorem \ref{Thm:RedBlue} that $pq$ leaves $\DT(P)$ at any moment during $(t_0,t_1)$. 
Nevertheless, suppose that $t_0$ is the time of a Delaunay co-circularity at which $pq$ leaves $\DT(P)$, and $t_1$ is the first time after $t_1$ when $pq$ re-enters $\DT(P)$.
Then Theorem \ref{Thm:RedBlue} relates such Delaunay co-circularities to $k$-shallow collinearities and co-circularities which occur in $\A_{pq}$ when the edge $pq$ under consideration is not Delaunay.
Therefore, this theorem can be regarded, in its own right, as one of the main contributions of this paper.


The proof of Theorem \ref{Thm:RedBlue} is based on the following simple idea. Assume that the edge $pq$ does not belong to $\DT(P)$ at a fixed time $t\in (t_0,t_1)$. If the Delaunayhood of $pq$ is violated by $\Omega(k^2)$ red-blue pairs $(r,b)$, then we encounter, during $(t_0,t_1)$, $\Omega(k^2)$ co-circularities (each involving $p,q$ and the corresponding pair $r,b$), or at least $\Omega(k)$ points $r$ change their color there by crossing $\L_{pq}$. 
Finally, if the Delaunayhood of $pq$ is violated at time $t$ by only $O(k^2)$ pairs, then it can be restored by removing a subset $A\subset P\setminus\{p,q\}$ of cardinality at most $O(k)$.

Impatient readers may safely skip the full proof of Theorem \ref{Thm:RedBlue}, which involves a fairly routine 
planar analysis in the above arrangement $\A_{pq}$ of red and blue curves. (A very similar argument was used in \cite{ASS} to address a totally different problem.)

\paragraph{Proof of Theorem \ref{Thm:RedBlue}.}
Without loss of generality, we assume that the edge $pq$ is Delaunay at time $t_0$. (If $pq$ is Delaunay at time $t_1$ then we can argue in a fully symmetrical fashion.)

Consider the portion of the red-blue arrangement associated with $pq$ within the time interval $(t_0,t_1)$. As above, refer to the parametric plane in which this arrangement is represented as the $t\rho$-plane, where $t$ is the time axis and $\rho$ measures signed distances  from $\L_{pq}$.
We define the {\it red} (resp., {\it blue}) {\it level} of a point $x=(t,\rho)$ in this parametric $\reals^2$ as the number of red (resp., blue) functions that lie below (resp., above) $x$ (in the $\rho$-direction). See Figure \ref{Fig:RedBlueLevels}.
It is easily checked that the level of a co-circularity event at time $t$, with circumcenter at distance $\rho$ from $\L_{pq}$, is the sum of the red and the blue levels of $(t,\rho)$.

\begin{figure}[htbp]
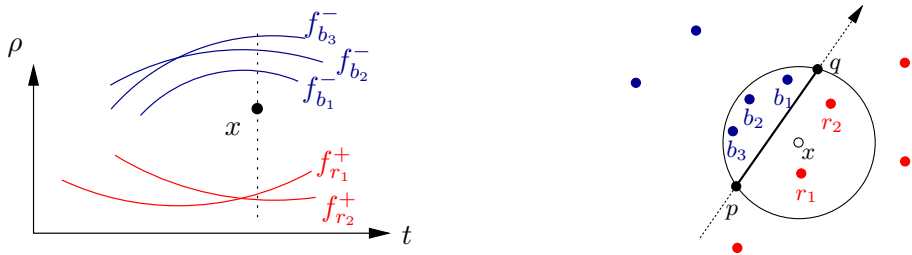

\begin{center}
\input{RedBlueLevels.pstex_t}\hspace{3cm}\input{RedBlueLevels1.pstex_t}
\caption{\small Left: The point $x=(t,\rho)$ lies below three blue functions and above two red functions, so its blue and red levels are $3$ and $2$, respectively. Right: The circumdisc centered at (signed) distance $\rho$ from $\L_{pq}$ and touching $p$ and $q$ at time $t$ contains the three corresponding blue points and two red points.}
\label{Fig:RedBlueLevels}
\end{center}
\end{figure} 

We distinguish between the following (possibly overlapping) cases:

\smallskip
\noindent {\bf (a)} $p$ and $q$ participate in a $k$-shallow collinearity with a third point $r$ at some moment during $I$. That is, condition (i) is satisfied. (Note that here we do not care whether $r$ crosses $pq$ or $\L_{pq}\setminus pq$.)

Suppose that this does not happen. That is, each time when a point $r\in P$ changes its color from red to blue or vice versa, the number of points on each side of $\L_{pq}$ is larger than $k$. Hence, either the number of points on each side of $\L_{pq}$ is always larger than $k$ (during $(t_0,t_1)$), or the sets of red and blue points remain fixed throughout $(t_0,t_1)$ (no crossing takes place), and the size of one of them is at most $k$. More concretely, either one of the sets contains fewer than $k$ points at the start of $I$, and then no crossing can ever occur during $I$, or both sets contain at least $k$ points at the start of $I$, and this property is maintained during $I$, by assumption. In the former case condition (iii) trivially holds, since removal of all points in $P\cap \L_{pq}^+$ or in $P\cap \L_{pq}^-$ guarantees that $pq$ is a hull edge throughout $(t_0,t_1)$, and thus belongs to the Delaunay triangulation. Hence, we may assume that the number of red points, and the number of blue points, are always both larger than $k$ during $(t_0,t_1)$.

\begin{figure}[htbp]
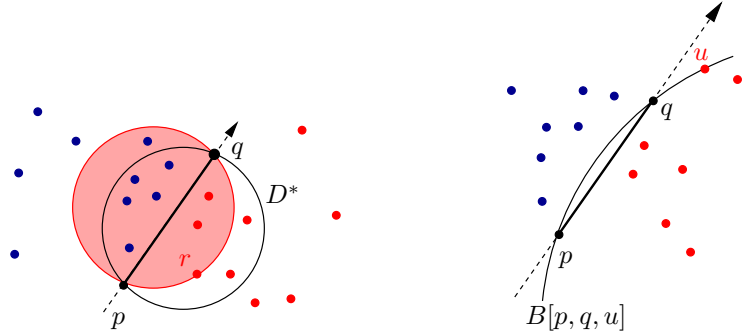

\begin{center}
\input{DeepDeep.pstex_t} \hspace{2cm} \input{CrossOutside.pstex_t}
\caption{\small Left: Case (b). The disc $D^*$ contains at least $k=5$ red points, and at least $k$ blue points. If $r$ lies at red level at most $k$, it belongs to $D^*$. Hence, the circumdisc $B[p,q,r]$ contains at least $k$ blue points, so the blue level of $f^+_r$ is at least $k$. Right: Case (c). The setup right after time $t'$ when $u$ crosses $\L_{pq}\setminus pq$. $B[p,q,u]$ contains at least $3k$ red points and no blue points.}
\label{Fig:DeepDisc}
\end{center}
\end{figure} 

\smallskip
\noindent {\bf (b)} At some moment $t_0\leq t^*\leq t_1$ there is a disc $D^*$ that touches $p$ and $q$, and contains at least $\lceil k/3\rceil$ red points and at least $\lceil k/3\rceil$ blue points.
In particular, for each of the $\lceil k/3\rceil$ shallowest red functions $f_r^+$ at time $t^*$, its respective red point $r$ belongs to $D^*$ and similarly for the $\lceil k/3\rceil$ shallowest blue functions.
See Figure \ref{Fig:DeepDisc} (left). Before we use the existence of $D^*$ we first conduct the following structural analysis.

Let $f_r^+$ be a red function which is defined at time $t_0$, and whose red level is then at most $\lfloor k/6\rfloor$. (Recall that, at time $t_0$, the blue level of any red function is $0$ since $pq$ belongs to $\DT(P)$.) We claim that either $f_r^+$ is defined and continuous throughout $(t_0,t_1)$ and its red level is always at most $\lceil k/3\rceil$, or $r$ participates in at least $\lceil k/6\rceil$ red-red and/or red-blue co-circularities, all of which are $\lceil k/3\rceil$-shallow. 

Indeed, the circumdisc $B[p,q,r]$ contains at most $\lfloor k/6\rfloor$ red points (and no blue points) at time $t_0$, and it moves continuously as long as $r$ remains in $\L_{pq}^+$. By the time at which either (the graph of) $f_r^+$ reaches red level $\lceil k/3\rceil$ or $r$ hits $\L_{pq}$, this disc ``swallows" either at least $\lceil k/6\rceil$ red points (either in the former case or in the latter case when $r$ crosses $\L_{pq}\setminus pq$) or at least $\lceil k/6\rceil$ blue points (in the latter case when $r$ crosses $pq$). (Recall that, by assumption, the number of red points and the number of blue points is always larger than $k$ during $I$.) We thus obtain at least $\lceil k/6\rceil$ $\lceil k/3\rceil$-shallow red-red or red-blue co-circularities involving $p,q,r,$ and a fourth (red or blue) point. 

To recap, if at least $\lfloor k/12\rfloor$ red functions, which at time $t_0$ are among the $\lceil k/6\rceil$ shallowest red functions, reach red level at least $\lceil k/3\rceil+1$, or have a discontinuity at $\rho=-\infty$ or $+\infty$ (at a crossing of $\L_{pq}$ by the corresponding point), then we encounter $\Omega(k^2)$ co-circularities (involving $p$ and $q$) which are $k$-shallow, so condition (ii) holds.

Hence, we may assume that at least $\lceil k/12 \rceil$ red functions $f_r^+$ that are among the $\lceil k/6\rceil$ shallowest red functions at time $t_0$, are defined throughout $(t_0,t_1)$, and their red level always remains at most $\lceil k/3\rceil$.
Fix any such red function $f_r^+$.
Clearly, the red point $r$ that defines $f_r^+$ belongs to $D^*$ at time $t^*$, and the circumdisc $B[p,q,r]$ contains at least $\lceil k/3\rceil$ blue points. See Figure \ref{Fig:DeepDisc} (left).
This implies that the blue level of $f_r^+$ reaches $\lceil k/3\rceil$ so (since the blue level was $0$ at time $t_0$) $r$ participates in at least $\lfloor k/6\rfloor$ $\lceil k/3\rceil$-shallow co-circularities during $(t_0,t^*)$. Repeating this argument for each of the remaining $\lceil k/12\rceil$ such red functions, we conclude that condition (ii) is again satisfied.

\smallskip
\noindent {\bf(c)} Suppose that neither of the two cases (a), (b) holds. 
Let $A_R$ (resp., $A_B$) be the subset of all points $u$ whose red (resp., blue) functions $f_u^+$ (resp., $f^-_u$) appear at red (resp., blue) level at most $\lceil k/3 \rceil$ at some moment during $(t_0,t_1)$. 

Since the situation in (b) does not occur, we can restore the Delaunayhood of $pq$, throughout the entire interval $(t_0,t_1)$, by removing all points in $A_R\cup A_B$. To see this, suppose that $pq$ is not Delaunay (in $\DT(P\setminus (A_R\cup A_B))$) at some time $t_0< t^*< t_1$. This is witnessed by a disc $D^*$ whose boundary passes through $p$ and $q$ and which contains a red point $r\not\in A_R$ and a blue point $b\not\in A_B$. Since the red level of $f_r^+$ is greater than $\lceil k/3\rceil$ at time $t^*$, $D^*$ must also contain the $\lceil k/3\rceil$ red points corresponding to the $\lceil k/3\rceil$ shallowest red functions at time $t^*$, and, symmetrically, also the $\lceil k/3\rceil$ blue points corresponding to the $\lceil k/3 \rceil$ shallowest blue functions at time $t^*$. But then $D^*$ satisfies the condition (b), contrary to assumption.

Let $A_R^o$ (resp., $A_B^o$) be the set of $k$ points whose red (resp., blue) functions are shallowest at time $t_0$.
It remains to consider the case where at least $k$ points $u$ in $A_R\cup A_B$ belong to neither of $A_R^o,A_B^o$, for otherwise condition (iii) is trivially satisfied, with a removed set of size at most $3k$.
Fix such a point $u$ and consider the first time $t^*\in (t_0,t_1)$ when its red function $f_u^+$ has red level at most $\lceil k/3\rceil$, or its blue function $f_u^-$ has blue level at most $\lceil k/3\rceil$. Without loss of generality, suppose that at time $t^*$ the red function $f_u^+$ has red level at most $\lceil k/3\rceil$.
We claim that $u$ does not cross $pq$ during $(t_0,t^*]$. Indeed, if there were such a crossing from $\L_{pq}^-$ to $\L_{pq}^+$ then the blue function $f_u^-$ would tend to $\infty$ right before the crossing, and its blue level would then be $0$ even before $t^*$, contrary to the choice of $t^*$. Similarly, if the crossing were from $\L_{pq}^+$ to $\L_{pq}^-$ then the red level of $f_u^+$ would be $0$ just before the crossing, again contradicting the choice of $t^*$. 

First, assume that $u$ does not cross $\L_{pq}$ during $(t_0,t^*)$, so the graph of $f_u^+$ is continuous during this time interval. Hence, the motion of the circumdisc $B[p,q,u]$ is also continuous.
Since $u\not\in A_R^o$, at time $t_0$ the circumdisc $B[p,q,u]$ contains at least $k$ red points and no blue points. At time $t^*$,
$B[p,q,u]$ contains $\lceil k/3\rceil$ red points and fewer than $\lceil k/3\rceil$ blue points (otherwise Case (b) would occur). 
Hence, we encounter at least $\lfloor k/3\rfloor$ $k$-shallow co-circularities during $(t_0,t^*)$, each involving $p,q,u$ and some other point of $P$.

Now, suppose $u$ crosses $\L_{pq}\setminus pq$ during $(t_0,t^*)$, and consider the last time $t'$ when this happens.
We can use exactly the same argument as in the ``continuous" case but now starting from $t'$. Indeed, $f_u^+$ is continuous during $(t',t^*]$ and, right after $t'$, the circumdisc $B[p,q,u]$ contains (all the red points and thus) at least $k$ red points, and no blue points. See Figure \ref{Fig:DeepDisc} (right).

Repeating this argument for all such points $u\in A_R\cup A_B\setminus(A_R^o\cup A_B^o)$, we get $\Omega(k^2)$ $k$-shallow co-circularities which occur during $(t_0,t_1)$ and involve $p$ and $q$. Hence, condition (ii) is again satisfied. This completes the proof of Theorem \ref{Thm:RedBlue}. $\Box$

\paragraph{Combinatorial charging schemes.}
To conclude this section, we briefly review the following general paradigm, which is widely used in computational geometry to bound the combinatorial complexity of various substructures in arrangements of (mostly non-linear) objects; see, e.g., \cite{Envelopes3D,ConstantLines} and \cite[Section 7]{SA95}.

Suppose that we are given two collections $\X$ and $\Y$ of geometric configurations, and wish to upper-bound the cardinality $|\X|$ of $\X$ in terms of the cardinality $|\Y|$ of $\Y$.
Note that the configurations in $\X$ and $\Y$ are usually of different types. For example, $\X$ and $\Y$ can consist, respectively, of Delaunay co-circularities and of $k$-shallow collinearities.

The most elementary class of charging schemes (which we shall use throughout this paper) is prescribed by a function $\lambda$ which maps each element $x\in \X$ to a subset $\lambda(x)\subseteq \Y$. 
We then say that that every element $y\in \lambda(x)$ is {\it charged} by $x$.
We also say that a configuration $y\in \Y$ is charged $\beta_y$ times if $\X$ contains exactly $\beta_y$ configurations $x$ whose respective images $\lambda(x)$ contain $y$. Furthermore, we say that $y\in \Y$ is charged {\it uniquely} if there is exactly one $x\in \X$ whose image $\lambda(x)$ contains $y$ (so $x$ is {\it uniquely determined} by the choice of $y$).

The resulting relation between $|\X|$ and $|\Y|$ depends on the following two parameters $\alpha$ and $\beta$ associated with our charging rule $\lambda$. The first parameter $\alpha$ denotes the minimum cardinality $|\lambda(x)|$ of $\lambda(x)$ (over all possible choices of $x\in \X$).
 The second parameter $\beta$ denotes the maximal possible number $\beta_y$ of configurations $x\in \X$ whose images contain a given configuration $y\in \Y$ (where the maximum is taken over all choices of $y\in \Y$).
In other words, $\alpha$ denotes the minimum number of configurations in $\Y$ that can be charged by the same  $x\in \X$, and $\beta$ denotes the maximum number of configurations $x\in \X$ that can charge the same $y\in \Y$.
With the above definitions, a standard double counting argument immediately shows that $|\X|\leq \frac{\beta |\Y|}{\alpha}$.

Therefore, in order to obtain the best possible upper estimate of $|\X|$, we seek to maximize $\alpha$, and to minimize 
$\beta$. 
In all our charging schemes, the mapping $\lambda$ will be constructed explicitly, so the value of $\alpha$ will be clear from the construction (and, most often, equal to $1$, with one significant exception). Thus, the main challenge will be to keep the value of $\beta$ under control (i.e., make sure that each configuration $y\in \Y$ is charged by relatively {\it few} members of $\X$).

\section{The Number of Delaunay Co-circularities}\label{Sec:DelCocircs}
In what follows, 
we assume that any four points in the underlying set $P$ are co-circular at most {\it twice} during their pseudo-algebraic motion. 
In this section we show that the maximum possible number $N(n)$ of Delaunay co-circularities in a set $P$, as above, is asymptotically dominated (if it is at least super-quadratic) by
the number of certain carefully defined configurations which will be referred to as {\it Delaunay crossings}. The analysis of Delaunay crossings will be postponed to Section \ref{Sec:CrossOnce}, where we shall impose additional restrictions on the collinearities that can be performed by triples of points in $P$.

\smallskip
\noindent{\it Definition.} We say that a co-circularity event at time $t_0$ involving $a,b,p,q$ has {\it index} $1$ (resp., $2$) if this is the first (resp., second) co-circularity involving $a,b,p,q$. 

\smallskip
To bound the maximum possible number of Delaunay co-circularities in $P$, we fix one such event at time $t_0$, at which an edge $pq$ of $\DT(P)$ is replaced by another edge $ab$, because of a red-blue co-circularity (with respect to $pq$, and, for that matter, also with respect to $ab$) of level $0$.
Assume first that the co-circularity of $p,q,a,b$ has index $2$; the case of index $1$ is handled fully symmetrically, by reversing the direction of the time axis.
 
There are at most $O(n^2)$ such events for which the vanishing edge $pq$ never reappears in $\DT(P)$, so we focus on the Delaunay co-circularities (of index $2$) whose corresponding edge $pq$ rejoins $\DT(P)$ at some future moment $t_1>t_0$ (or right after it).

Specifically, $\DT(P)$ experiences at time $t_1$ either a Delaunay co-circularity or a hull event (at which $pq$ is hit by some point of $P\setminus \{p,q\}$). In the latter case, $rq$ is not strictly Delaunay at time $t_1$ and appears in $\DT(P)$ only {\it right after} this event. 

Note that in this case, if the co-circularity at time $t_0$ involved two other points $a,b$, then at least one of $a,b$ must cross $\L_{pq}$ between $t_0$ and $t_1$ otherwise $p,q,a$ and $b$ would have to become co-circular again, in order to ``free" $pq$ from non-Delaunayhood, which is impossible since our co-circularity is assumed to be {last} co-circularity of $p,q,a,b$.

More generally, we have the following topological lemma:
\begin{lemma}\label{Lemma:MustCross}
Assume that the Delaunayhood of $pq$ is violated at time $t_0$ (or rather right after it) by the points $a\in \L_{pq}^-$ and $b\in \L_{pq}^+$. 
Furthermore, suppose that $pq$ enters $\DT(P)$ at some future time $t_1>t_0$.
Then at least one of the followings occurs during $(t_0,t_1]$: 

\medskip
(1) The point $a$ crosses $pq$ from $\L_{pq}^-$ to $\L_{pq}^+$.\\
\indent(2) The point $b$ crosses $pq$ from $\L_{pq}^+$ to $\L_{pq}^-$.\\
\indent(3) The four points $p,q,a,b$ are involved in another co-circularity (which is also red-blue with respect to $pq$). 
\end{lemma}

A symmetric version of Lemma \ref{Lemma:MustCross} applies if the Delaunayhood of $pq$ is violated at time $t_0$ (or right before it) by $a$ and $b$, and this edge is Delaunay at an {\it earlier} time $t_1<t_0$.

\begin{proof}
Refer to Figure \ref{Fig:StaysViolated}.
Clearly, the Delaunayhood of $pq$ remains violated by $a$ and $b$ after time $t_0$ as long as
$a$ remains within the cap $B[p,q,b]\cap \L_{pq}^-$, and $b$ remains within the cap $B[p,q,a]\cap \L_{pq}^+$ (as depicted in the left figure). 

Consider the first time $t^*\in (t_0,t_1]$ when the above state of affairs ceases to hold. Notice that $pq$ is intersected by $ab$ throughout the interval $[t_0,t^*)$.
Assume without loss of generality that $a$ leaves the the cap $B[p,q,b]\cap \L_{pq}^-$. If $a$ crosses $pq$, then the first scenario holds. Otherwise, $a$ can leave the above cap only through the boundary of $B[p,q,b]$ (as depicted in the right figure), so the third scenario occurs.
\end{proof}
\begin{figure}[htb]
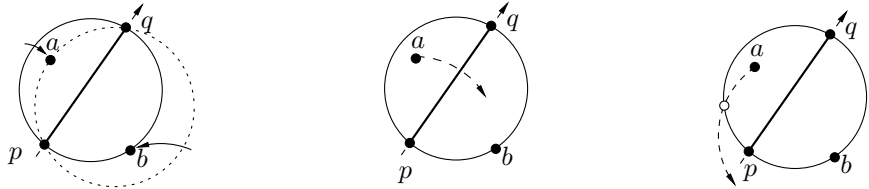

\begin{center}
\input{AfterLastCocirc.pstex_t}\hspace{2.5cm}\input{CrossWithin.pstex_t}\hspace{2.5cm}\input{NotOutside.pstex_t}
\caption{\small Proof of Lemma \ref{Lemma:MustCross}. Left: The setup right after time $t_0$. Center and right: the point $a$ can leave $B[p,q,b]$ in two possible ways.}
\label{Fig:StaysViolated}
\end{center}
\vspace{-0.3cm}
\end{figure}


Notice, however, that the points of $P$ can define $\Omega(n^3)$ collinearities, so a naive charging of Delaunay co-circularities to collinearities of type (1) or (2) in Lemma \ref{Lemma:MustCross} will not lead to a near-quadratic upper bound. (In other words, the universe of {\it all} collinearity events is far too large for our purposes.) Therefore, before we get to charging collinearities, we perform several preliminary charging steps, which will account for some Delaunay co-circularities of index $2$ (thus removing their corresponding collinearities from consideration). 

As a preparation, we fix a constant parameter $k>12$, apply Theorem \ref{Thm:RedBlue} to the edge $pq$ over the interval $(t_0,t_1)$ of its absence from $\DT(P)$. We distinguish between three possible alternatives provided by that theorem.

\smallskip
\noindent{\bf (i)} If the first condition of Theorem \ref{Thm:RedBlue} is satisfied, we can charge the co-circularity of $p,q,a,$ and $b$ to a $k$-shallow collinearity that occurs in $(t_0,t_1)$ and involves $p,q,$ and some third point of $P$. As argued in Section \ref{Sec:Prelim}, the overall number of $k$-shallow collinearities is $O(kn^2\beta(n))$.

Clearly, any collinearity event is charged at most a constant number of times. Namely, it can be charged only for the disappearances of edges $pq$ whose two vertices $p,q$ participate in the event, and only for the disappearance immediately preceding the event, without any in-between reappearance.

To conclude, the number of Delaunay co-circularities that fall into case (i) of Theorem \ref{Thm:RedBlue} does not exceed $O(kn^2\beta(n))$.

\smallskip
\noindent{\bf (ii)} If the second condition of Theorem \ref{Thm:RedBlue} is satisfied, then we charge the Delaunay co-circularity at time $t_0$ to $\Omega(k^2)$ $k$-shallow co-circularities, each occurring in $(t_0,t_1)$ and involving $p,q,$ together with some two other points of $P$.

As argued in Section \ref{Sec:Prelim}, the overall number of $k$-shallow co-circularities is $O(k^4 N(n/k))$. Once again, each $k$-shallow co-circularity is charged by only $O(1)$ Delaunay co-circularities in this manner, because $t_0$ is the last disappearance of $pq$ before the charged event. Hence, at most $O(k^2N(n/k))=O\left(\frac{k^4}{k^2}N(n/k)\right)$ Delaunay co-circularities can fall into this case.

The above two cases account for at most $O(k^2N(n/k)+kn^2\beta(n))$ Delaunay co-circularities (of index $2$). If left to themselves, they would result in a recurrence of $N(n)=O(k^2N(n/k)+kn^2\beta(n))$, with a nearly quadratic solution (see below for details, and see, e.g., \cite{ASS} for similar situations).
Unfortunately, this scheme does not always work because there might exist Delaunay co-circularities for which the respective red-blue arrangement (of the disappearing edge $pq$) contains relatively few $k$-shallow co-circularities, and no $k$-shallow collinearities.
Such instances fall into the third case of Theorem \ref{Thm:RedBlue}, which is far more complicated to handle.

\smallskip
\noindent{\bf (iii)} There is a set $A$ of at most $3k$ points (necessarily including at least one of $a$ or $b$) whose removal ensures the Delaunayhood of $pq$ throughout $(t_0,t_1)$. Recall that, by Lemma \ref{Lemma:MustCross}, at least one the two points $a,b$, let it be $a$, crosses $pq$ during $(t_0,t_1]$. In the reduced triangulation $\DT(P\setminus A\cup\{a\})$, the collinearity of $p,q$ and $a$ is of a special type, and we refer to it as a {\it Delaunay crossing}. 

\medskip
\noindent{\bf Delaunay crossings.}
A {\it Delaunay crossing} is a triple $(pq,r,I=[t_0,t_1])$, where $p,q,r\in P$ and $I$ is a time interval, such that 
\begin{enumerate}
\item $pq$ 
leaves $\DT(P)$ at time $t_0$, and returns at time $t_1$ (and $pq$ does not belong to $\DT(P)$ during $(t_0,t_1)$),
\item $r$ crosses the segment $pq$ {\it at least} once during $I$, and
\item $pq$ is an edge of $\DT(P\setminus \{r\})$ during $I$ (i.e., removing $r$ restores the Delaunayhood of $pq$ during the entire time interval $I$).
\end{enumerate}

\begin{figure}[htbp]
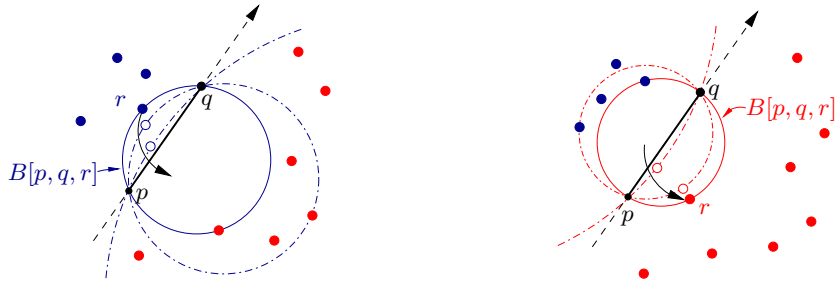

\begin{center}
\input{DelaunayCrossing.pstex_t}\hspace{3cm}\input{DelaunayCrossing1.pstex_t}
\caption{\small A Delaunay crossing of $pq$ by $r$ from $\L_{pq}^-$ to $\L_{pq}^+$. Several snapshots of the continuous motion of $B[p,q,r]$ before and after $r$ crosses $pq$ are depicted (in the left and right figures, respectively). Hollow points specify the positions of $r$ when $pq\not \in \DT(P)$. The solid circle in the left (resp., right) figure is the Delaunay co-circularity that starts (resp., ends) $I$.}
\label{Fig:DelaunayCrossing}
\end{center}
\end{figure} 

Note that we allow Delaunay crossings, where the point $r$ hits $pq$ at one (or both) of the times $t_0,t_1$. In this case, the crossed edge $pq$ leaves the convex hull of $P$ at time $t_0$, or enters it at time $t_1$. Clearly, the overall number of such ``degenerate" crossings is bounded by $O(n^2\beta(n))$.

It is easy to see that the third condition is equivalent to the following condition, expressed in terms of the red-blue arrangement $\A_{pq}$ associated with $pq$: The point $r$ participates only in red-blue co-circularites during the interval $I$, and these are the only red-blue co-circularities that occur during $I$.

More specifically, note that $r$ is red during some portion of $I$ and is blue during the complementary portion (both portions are nonempty, unless $r$ hits $pq$ when $I$ begins or ends). During the former portion the graph of $f_r^+$ coincides with the red lower envelope $E^+$ (otherwise $E^+(t)<E^-(t)$ would hold sometime during $I$ even after removal of $r$), so it can only meet the graphs of blue functions. Similarly, during the latter portion $f_r^-$ coincides with the blue upper envelope $E^-$, so it can only meet the graphs of red functions.
See Figure \ref{Fig:DelaunayCrossing} for a schematic illustration of this behavior.

Notice that no points, other than $r$, cross $pq$ during $I$ (any such crossing would clearly contradict the third condition at the very moment when it occurs).
Moreover, $r$ does not cross $\L_{pq}$ outside $pq$ during $I$; otherwise $pq$ would belong to $\DT(P)$ when $r$ belongs to $\L_{pq}\setminus pq$.


\medskip
\noindent{\bf Back to case (iii).} We can now express the number of remaining Delaunay co-circularities of index $2$ in terms of the maximum possible number of Delaunay crossings. To achieve this, we again resort to a probabilistic argument, in the spirit of Clarkson and Shor. Recall that for each such co-circularity there is a set $A$ of at most $3k$ points whose removal restores the Delaunayhood of $pq$ throughout $[t_0,t_1]$. In addition, we assume that $a$ hits $pq$ during $(t_0,t_1]$, and then $a\in A$.

We sample at random (and without replacement) a subset $R\subset P$ of $O(n/k)$ points, and notice that the following two events occur simultaneously with probability at least $\Omega(1/k^3)$: (1) the points $p,q,a$ belong to $R$, and (2) none of the points of $A\setminus\{a\}$ belong to $R$. 

Since $a$ crosses $pq$ during $[t_0,t_1]$, and $pq$ is Delaunay at time $t_0$ and (right after) time $t_1$, the sample $R$ induces a Delaunay crossing $(pq,a,I)$, for some time interval $I\subset [t_0,t_1]$. (If $a$ crosses $pq$ more than once, there may be several such crossings which occur at disjoint  sub-intervals of $[t_0,t_1]$, but it may also be the case that all these crossing form a single Delaunay crossing, in the way it was defined above. This depends on whether $pq$ manages to become Delaunay in $\DT(R)$ in between these crossings.)

We charge the disappearance of $pq$ from $\DT(P)$ to the above crossing in $R$ (or to the first such crossing if there are several) and note that the charging is unique (i.e., every Delaunay crossing $(pq,a,I)$ in $\DT(R)$ is charged by at most one disappearance of the respective edge $pq$ from $\DT(P)$).
Hence, the number of Delaunay co-circularities of this kind is bounded by $O(k^3 C(n/k))$, where $C(n)$ denotes the maximum number of Delaunay crossings induced by any collection $P$ of $n$ points whose motion satisfies the above assumptions.

If the Delaunay co-circularity of $p,q,a,b$ has index $1$, we reverse the direction of the time axis and argue as above for the edge $ab$ instead of $pq$. We thus obtain the following recurrence for the number of Delaunay co-circularities:

\begin{equation}\label{Eq:FirstRecurrence}
N(n)\leq c\left(k^2N(n/k)+k^3C(n/k)+kn^2\beta(n)\right),
\end{equation}
for some absolute constant $c>0$ which is independent of $k$.

Informally, (\ref{Eq:FirstRecurrence}) implies that the maximum number of Delaunay co-circularities is asymptotically dominated by the maximum number of Delaunay crossings.

\paragraph{Discussion.} In the above analysis, we have used Theorem \ref{Thm:RedBlue} for the edge $pq$, which vanishes at the Delaunay co-circularity, in order to decompose the universe of all such events into three sub-classes (which correspond to the respective three cases of the theorem). Within each sub-class of Delaunay co-circularities, we have devised an entirely different charging scheme. In all cases, the (almost-)uniqueness of charging has been guaranteed through the careful choice of the interval $(t_0,t_1)$, over which Theorem \ref{Thm:RedBlue} has been applied. Additional applications of this paradigm can be found in Section \ref{Sec:CrossOnce}.

\paragraph{The number of Delaunay co-circularities--wrap-up.}
In Section \ref{Sec:CrossOnce} we shall obtain the following recurrence for the maximum number $C(n)$ of Delaunay crossings:

\begin{equation}\label{Eq:AllCrossings}
C(n)\leq c_1\left(k_1^2 N(n/k_1)+k_1 k_2^2 N(n/k_2)+k_1k_2n^2\beta(n)\right), 
\end{equation}
where $k_1$ and $k_2$ are any two constants that satisfy $12<k_1\ll k_2$, and
$c_1>0$ is another constant which is independent of $k_1,k_2$.

Our analysis will rely on the following additional assumption on the pseudo-algebraic motions of $P$ (which was not necessary to establish (\ref{Eq:FirstRecurrence})):

\smallskip
{\it Either (i) no triple of points can be collinear more than twice, or (ii) no ordered triple of points can be collinear more than once}. 

\smallskip
Substituting the inequality (\ref{Eq:AllCrossings}) into (\ref{Eq:FirstRecurrence}), we obtain the following recurrence for $N(n)$, in which we choose $k\ll k_1\ll k_2$: 
\begin{equation}\label{Eq:FinalRecurrence}
N(n)\leq c_2\left(k^2 N\left(\frac{n}{k}\right)+k^3 k_1^2N\left(\frac{n}{k_1k}\right)+k^3k_1k_2^2N\left(\frac{n}{k_2k}\right)+kk_1k_2n^2\beta(n)\right),
\end{equation}
where $c_2$ is a constant factor which does not depend on the choice of $k,k_1,k_2$.

Arguing as in earlier solutions of similar charging-based recurrences (see, e.g., \cite{Envelopes3D,ConstantLines}, or \cite[Section 7.3.2]{SA95}), the recurrence solves to
$N(n)=O(n^{2+\eps})$, for any $\eps>0$. (Specifically, for a given $\eps>0$, we choose the parameters $k\ll k_1\ll k_2$ as functions of $\eps>0$, and establish the bound $O(n^{2+\eps})$ with a constant of proportionality depending on $\eps$, using induction on $n$.) 

In conclusion, we have the following main result of this paper.

\begin{theorem}\label{Thm:OverallBound}
Let $P$ be a collection of $n$ points moving along pseudo-algebraic trajectories in the plane, so that any four points of $P$ are co-circular at most twice. Assume also that either (i) no triple of points can be collinear more than twice, or (ii) no {\it ordered} triple of points can be collinear more than once. Then the Delaunay triangulation $\DT(P)$ of $P$ experiences at most $O(n^{2+\eps})$ discrete changes throughout the motion, for any $\eps>0$.
\end{theorem}

\section{The Number of Delaunay Crossings}\label{Sec:CrossOnce}

In this section we complete the proof of Theorem \ref{Thm:OverallBound}. 
Throughout this section, we assume that any four points in the underlying set $P$ of $n$ moving points are {\it co-circular at most twice}, and that either {\it (i) no triple of points can be collinear more than twice}, or {\it (ii) no ordered triple of them can be collinear more than once}.\footnote{The last condition (ii) is equivalent to the following one: There can be at most one collinearity of an ordered triple $(p,q,r)$ at which $r$ hits $pq$.}
With these assumptions, we show that the maximum possible number $C(n)$ of Delaunay crossings in any set $P$ as above satifies
the recurrence relation (\ref{Eq:AllCrossings}) asserted in the end of Section \ref{Sec:DelCocircs}.

Let $(pq,r,I=[t_0,t_1])$ be a Delaunay crossing, as defined in the previous section. Specifically, $pq$ disappears from $\DT(P)$ at time $t_0$, rejoins $\DT(P)$ at time $t_1$, and remains Delaunay throughout $I$ in $\DT(P\setminus \{r\})$. To distinguish between the notion of a Delaunay crossing $(pq,r,I)$, which lasts for the full time interval $I$, and the instance where $r$ actually lies on the segment $pq$, we refer to the latter event by saying that $r$ {\it hits} $pq$.

\paragraph{Types of Delaunay crossings.}
Notice that $r$ can hit the edge $pq$ at most twice during the above crossing $(pq,r,I)$, for otherwise the ordered triple $(p,q,r)$ will be collinear at least three times.

A Delaunay crossing $(pq,r,I=[t_0,t_1])$ is called {\it single} if the point $r$ hits $pq$ only once during $I$.
Otherwise (i.e., if $r$ hits $pq$ exactly twice during $I$), we say that $(pq,r,I=[t_0,t_1])$ is a {\it double} Delaunay crossing. 

Note that double Delaunay crossings can only arise if no three points in $P$ can be collinear more than twice. (That is, double crossings are simply {\it impossible} in the second setting, where no ordered triple in $P$ can be collinear more than once.)

In Section \ref{Subsec:Single} (namely, in Theorem \ref{Thm:OrdinaryCrossings}), we show that the maximum possible number $C_1(n)$ of single Delaunay crossings in the above set $P$ satisfies the following recurrence:

\begin{equation}\label{Eq:Single}
C_1(n)=O\left(k_1^2 N(n/k_1)+k_1 k_2^2 N(n/k_2)+k_1k_2n^2\beta(n)\right), 
\end{equation}
where $k_1$ and $k_2$ are any two constants that satisfy $12<k_1\ll k_2$, and the constant of proportionality in $O(\cdot)$ does not depend on $k_1,k_2$. Curiously enough, our analysis of single Delaunay crossings is equally valid given {\it any} of the two alternative assumptions (i), (ii) concerning the collinearities performed by the points of $P$.

In Section \ref{Subsec:Double} (namely, in Theorem \ref{Thm:SpecialCrossings}) we show that any set $P$ of $n$ points, whose pseudo-algebraic motions satisfy the above assumptions, admits at most $O(n^2)$ double Delaunay crossings. Specifically, we argue that any double Delaunay crossing $(pq,r,I)$ can be {\it uniquely} (or almost-uniquely) charged to one of its respective edges $pr,rq$.
In our analysis of double Delaunay crossings we can rely on the assumption that no three points of $P$ can be collinear more than twice, because otherwise such crossings do not arise at all.

The overall Recurrence (\ref{Eq:AllCrossings}) for $C(n)$, asserted in the end of Section \ref{Sec:DelCocircs}, will follow immediately by combining the above two bounds.


\smallskip
Both Sections \ref{Subsec:Single} and \ref{Subsec:Double} use the following simple lemma.

\begin{lemma}\label{Lemma:Crossing}
If $(pq,r,I=[t_0,t_1])$ is a Delaunay crossing then each of the edges $pr,rq$ belongs to $\DT(P)$ throughout $I$.
\end{lemma}
\paragraph{Remark:}
Most applications of Lemma \ref{Lemma:Crossing} (especially in Section \ref{Subsec:Single}) 
rely only on the fact that the edges $pr$ and $rq$
are Delaunay at times $t_0$ and $t_1$.
To establish the Delaunayhood of $pr$ and $rq$ at time $t_0$ it is sufficient to observe that, at that moment, there occurs a Delaunay co-circularity involving $p,q,r$ and some other point $s$; moreover, this co-circularity is red-blue with respect to $pq$. Hence, $\DT(P)$ contains the triangle $\triangle pqr$ right before $t_0$, so the edges $pr$, $rq$ are then Delaunay.\footnote{For degenerate crossings (which begin with a collinearity of $p,r,q$), the edge $pq$ is replaced on the convex hull of $P$ by $pr$ and $rq$.}
A symmetric argument shows that $pr$ and $rq$ are Delaunay at time $t_1$. 
The stronger form of the lemma is used mainly in Section \ref{Subsec:Double}.
\begin{proof}
We prove the claim only for the edge $rq$ and for $t\in I$ at which $r$ lies in $\L_{pq}^-$; the complementary portion of $I$, and the corresponding treatment of $pr$, are handled symmetrically. The crucial observation is that, during the chosen portion of $I$, $f_r^-(t)$ coincides with the blue upper envelope $E^-(t)$ (defined with respect to $pq$). Indeed, let $x\in P\cap \L_{pq}^+$ be any red point so that the Delaunayhood of $pq$ is violated at time $t\in I$ by $x$ and $r$. Then the Delaunayhood of $pq$ is also violated there by $x$ and any blue point $y\in P\cap \L_{pq}^-$ whose respective blue function $f_y^-(t)$ coincides with $E^-(t)$, implying that $y=r$.
Therefore, the cap $B[p,q,r]\cap \L_{pq}^-$ has $P$-empty interior throughout the chosen portion of $I$. 

Suppose that $rq$ is not Delaunay at some time $t^*$ that belongs to the chosen portion of $I$.
We now consider the red-blue arrangement of $rq$ at that moment. Let $x\in P\cap \L_{rq}^+$ be the point whose function $f_x^+(t^*)$ coincides with the red lower envelope $E^+(t^*)$ (with respect to $qr$). In particular, we have $f^+_x(t^*)\leq f^+_p(t^*)$ (as is easily checked, $p\in \L_{rq}^+$, when $r\in \L_{pq}^-$). Clearly, $x$ cannot be equal to $p$, for, otherwise, the disc $B[p,q,r]$ would have $P$-empty interior. Indeed, we argued that $B[p,q,r]\cap \L_{pq}^-$ is $P$-empty, and a similar argument shows that $B[p,q,r]\cap \L_{rq}^+$ would also have to be empty if $x$ and $p$ coincide, from which the emptiness of the whole interior follows. It follows that $pq$ is Delaunay at time $t^*\in I$,
contradicting the definition of a Delaunay crossing. See Figure \ref{Fig:CrossingLemma}.
Moreover, $x$ cannot lie in $\L_{pq}^-$, for it would then have to lie in $B[p,q,r]\cap \L_{pq}^-$, which is impossible since this portion of $B[p,q,r]$ is $P$-empty.

\begin{figure}[htbp]
\begin{center}
\input{StayDelaunay.pstex_t}
\caption{\small Proof of Lemma \ref{Lemma:Crossing}.}
\label{Fig:CrossingLemma}
\end{center}
\end{figure}

Since $rq$ is not Delaunay, the disc $B=B[q,r,x]$ contains another point $y\in P\cap\L_{rq}^-$, which is easily seen to lie in $\L_{pq}^-$ and in $\L_{xq}^-$.
We can expand $B$ from $qx$ until its boundary touches $p$, $q$ and $x$, and its interior contains $y$. This implies that $pq$ does not belong to $\DT(P\setminus\{r\})$ at time $t\in I$, which contradicts the definition of a Delaunay crossing. 
\end{proof}


\subsection{The number of single Delaunay crossings}\label{Subsec:Single}
In this subsection we establish Recurrence (\ref{Eq:Single}) for the maximum possible number $C_1(n)$ of single Delaunay crossings in a set $P$ of $n$ points whose pseudo-algebraic motions satisfy the above assumptions.
To facilitate the proof of this main result, which is asserted in the culminating Theorem \ref{Thm:OrdinaryCrossings}, we begin by introducing some additional notation, and by establishing several auxiliary lemmas.

\paragraph{Notational conventions.}
Recall from Section \ref{Sec:Prelim} that every edge $pq$ is oriented from $p$ to $q$, and its corresponding line $\L_{pq}$ splits the plane into halfplanes $\L_{pq}^-$ and $\L_{pq}^+$.

Without loss of generality, we assume in what follows that, for any single Delaunay crossing $(pq,r,I=[t_0,t_1])$, the point $r$ crosses $pq$ from $\L_{pq}^-$ to $\L_{pq}^+$ during $I$. Recall that $r$ cannot cross $\L_{pq}$ outside $pq$ during $I$, so this is the {\it only} collinearity of $p,q,r$ in $I$.
If $r$ crosses $pq$ in the opposite direction, we regard this crossing as $(qp,r,I=[t_0,t_1])$.

Note that every such Delaunay crossing $(pq,r,I)$ is uniquely determined by the respective ordered triple $(p,q,r)$, because there can be at most one collinearity\footnote{If $r$ hits $pq$ twice, which is allowed only if no three points of $P$ can be collinear more than twice, then the other crossing of $pq$ by $r$ is from $\L_{pq}^+$ back to $\L_{pq}^-$.} where $r$ crosses the line $\L_{pq}$ {\it within $pq$} from $\L_{pq}^-$ to $\L_{pq}^+$.

For a convenience of reference, we label each such crossing $(pq,r,I)$ as {\it a clockwise $(p,r)$-crossing}, and as {\it a counterclockwise $(q,r)$-crossing}, with an obvious meaning of these labels.


The following lemma lies at the heart of our analysis.

\begin{lemma}\label{Lemma:OnceCollin}
Let $(pq,r,I=[t_0,t_1])$ be a Delaunay crossing. Then, with the above conventions, for any $s\in P\setminus\{p,q,r\}$ the points $p,q,r,s$ define a red-blue co-circularity with respect to $pq$, which takes place during $I$ when the point $s$ either enters the cap $B[p,q,r]\cap \L_{pq}^+$, or leaves the opposite cap $B[p,q,r]\cap \L_{pq}^-$.
\end{lemma}
\begin{proof}
By definition, $r$ crosses $pq$ at some (unique) time $t_0\leq t^*\leq t_1$, say from $\L_{pq}^-$ to $\L_{pq}^+$. The disc $B[p,q,r]$ is $P$-empty at $t_0$ and at $t_1$ and moves continuously throughout $[t_0,t^*)$ and $(t^*,t_1]$.
Just before $t^*$, $B[p,q,r]$ is the entire $\L_{pq}^+$, so every point $s\in P\cap \L_{pq}^+$ at time $t^*$ must have entered $B[p,q,r]$ during $[t_0,t^*)$, forming a co-circularity with $p,q,r$ at the time it enters the disc.\footnote{If $t^*=t_0$ then
there are no red points when $r$ hits $pq$, so we consider only the second interval. The case of $t^*=t_2$ is treated symmetrically.} See Figure \ref{Fig:BeforeCrossing} (left). (As mentioned in Section \ref{Sec:Prelim}, this co-circularity of $p,q,r,s$ is red-blue with respect to $pq$, that is, the point $s$ enters $B[p,q,r]$ through $\partial B[p,q,r]\cap \L_{pq}^+$.) A symmetric argument (in which we reverse the direction of the time axis) shows that the same holds for all the points $s\in P$ that lie in $\L_{pq}^-$ at time $t^*$; see Figure \ref{Fig:BeforeCrossing} (right). 
\end{proof}


\begin{figure}[htbp]
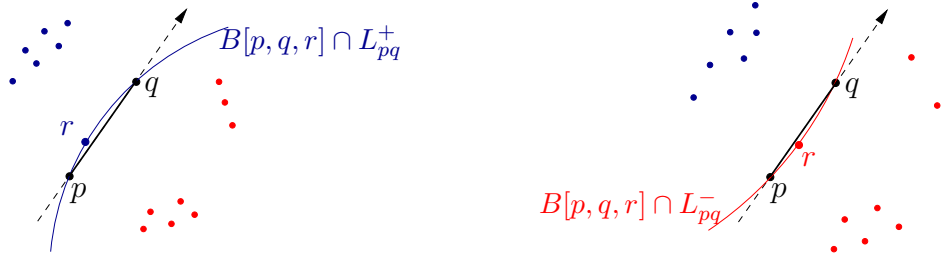

\begin{center}
\input{BeforeCrossing.pstex_t}\hspace{4cm}\input{AfterCrossing.pstex_t}
\caption{\small Left: Right before $r$ crosses $pq$, the circumdisc $B=B[p,q,r]$ contains all points in $P\cap \L_{pq}^+$. Right: After $r$ crosses $pq$, $B$ contains all points in $P\cap \L_{pq}^-$.}
\label{Fig:BeforeCrossing}
\end{center}
\end{figure} 


\begin{lemma}\label{Lemma:TwiceCollin}
The number of triples of points $p,q,r\in P$ for which there exist two time intervals $I_1,I_2$ such that both $(pq,r,I_1)$ and $(rq,p,I_2)$ are single Delaunay crossings, is at most $O(n^2)$. Furthermore, the lemma still holds if we reverse 
$pq$, or $rq$, or both. 
\end{lemma}

In other words, the lemma asserts that $P$ contains at most quadratically many triples $p,q,r$ that perform two single Delaunay crossings of distinct order types.

\begin{proof}
Assume, with no loss of generality, that the Delaunay crossing of $rq$ (or $qr$) by $p$ ends after the crossing of $pq$ (or $qp$) by $r$; that is $I_2$ ends after $I_1$ (note that $I_1$ and $I_2$ need not be disjoint).
Let $s$ be a point of $P\setminus\{p,q,r\}$.
By Lemma \ref{Lemma:OnceCollin}, the four points $p,q,r,s$ define a co-circularity event during $I_1$.
Similarly, the same four points $p,q,r,s$ define a co-circularity event during $I_2$. We claim that the above two co-circularities are distinct. 
Indeed, the former co-circularity of $p,q,r,s$ is red-blue with respect to the edge $pq$ (which is crossed by $r$ during $I_1$), so $p$ and $q$ are not adjacent in the co-circularity. On the other hand, the latter co-circularity is red-blue with respect to $rq$ (which is crossed by $p$ during $I_2$), so $q$ and $r$ are not adjacent in the co-circularity. However, both non-adjacencies cannot occur simultaneously in the same co-circularity, so these two co-circularities of $p,q,r,s$ must be distinct.

Hence, the points $p,q,r,s$
induce at least (by our assumption, exactly) two common co-circularity events before $rq$ re-enters $\DT(P)$. 

Thus, we cannot have a Delaunay crossing of any of $rq,qr$ by $s$ after $rq$ re-enters $\DT(P)$, for otherwise this would lead, according to Lemma \ref{Lemma:OnceCollin} to a third co-circularity event involving $p,q,r$ and $s$. 
Since this holds for every point $s\in P\setminus \{p,q,r\}$, the crossing of $rq$ (or $qr$) by $p$ is the last Delaunay crossing of $rq$ (or $qr$), so it can be charged uniquely to this edge. (Clearly, any two Delaunay crossings of the same edge $rq$ take place at disjoint time intervals.)
\end{proof}


\smallskip

Our overall strategy is to show that, for an average choice of $p,r\in P$, there exist only few $(p,r)$-crossings of a given orientation type (which can be either clockwise or counterclockwise).
In other words, we are to show that most single Delaunay crossings $(pq,r,I)$ can be almost-uniquely charged to either one of its edges $pr$ and $rq$. (As a matter of fact, it is sufficient that we can charge $(pq,r,I)$ to only {\it one} of its edges $pr,rq$. As explained in Section \ref{Subsec:Double}, this simple charging succeeds for all {\it double} Delaunay crossings.)

Unfortunately, there can be arbitrary many single $(p,r)$-crossings, of both orientation types. In such cases, we resort to more intricate charging arguments (see the proof of Theorem \ref{Thm:OrdinaryCrossings}). Note that the respective intervals $I$ and $J$ of any pair of such crossings, say $(pq,r,I)$ and $(pa,r,J)$, may overlap.
The following lemma defines a natural order on $(p,r)$-crossings of a given orientation (clockwise or counterclockwise).

\begin{lemma}\label{Lemma: OrderOrdinaryCrossings}
Let $(pq,r,I)$ and $(pa,r,J)$ be clockwise $(p,r)$-crossings, and suppose that $r$ hits $pq$ (during $I$) before it hits $pa$ (during $J$). Then $I$ begins (resp., ends) before the beginning (resp., end) of $J$. Clearly, the converse statements hold too. Similar statements also hold for pairs of counterclockwise $(p,r)$-crossings.
\end{lemma}

\begin{proof}
In the configuration considered in the main statement of the lemma, $r$ crosses $pq$ from $\L_{pq}^-$ to $\L_{pq}^+$, and it crosses $pa$ from $\L_{pa}^-$ to $\L_{pa}^+$.
We only prove the part of the lemma concerning the ending times of the crossings, because the proof about the starting times is fully symmetric (by reversing the direction of the time axis). The statement clearly holds if $I$ and $J$ are disjoint; the interesting situation is when they partially overlap.
Note that $r$ enters $\L_{pq}^+$ only once during the Delaunay crossing of $pq$ by $r$, namely, right after $r$ hits $pq$. Indeed, by assumption, $r$ cannot exit $\L_{pq}^+$ by crossing $pq$ again during $I$, and it cannot cross $\L_{pq}\setminus pq$ because at that time $pq$, which is Delaunay in $\DT(P\setminus\{r\})$, would be Delaunay also in the presence of $r$, contrary to the definition of a Delaunay crossing.
Hence, we may assume that $r$ still lies in $\L_{pq}^+$ when it hits $pa$ during the Delaunay crossing of that edge. Indeed, otherwise the crossing of $pq$ would by then be over, so the claim would hold trivially, as noted above. In particular, $\vec{pa}$ lies clockwise to $\vec{pq}$ at that time.

\begin{figure}[htbp]
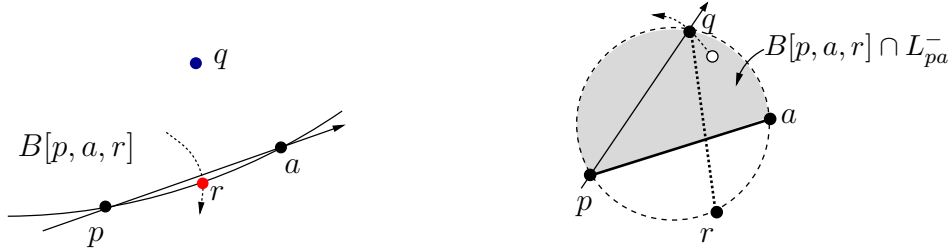

\begin{center}
\input{orderord1.pstex_t}\hspace{3cm}\input{OrderCrossings.pstex_t}
\caption{\small Proof of Lemma \ref{Lemma: OrderOrdinaryCrossings}. Left: if $r$ remains in $\L_{pq}^+$ after $I$ and before it crosses $pa$, then $q$ lies in $B[p,a,r]\cap\L_{pa}^-$ before that last collinearity. Right: The second co-circularity of $p,q,r,a$ which occurs when $q$ leaves $B[p,a,r]\cap \L_{pa}^-$. This is a red-red co-circularity with respect to $pq$, so the crossing of $pq$ is already over.}
\label{Fig:OrderOrdinaryCrossings}
\end{center}
\end{figure} 

It suffices to prove that the co-circularity of $p,q,r,a$, which (by Lemma \ref{Lemma:OnceCollin}) occurs during the Delaunay crossing of $pa$ by $r$, takes place when the crossing of $pq$ by $r$ is already finished (and, in particular, after the co-circularity of $p,q,r,a$ that occurs during the crossing of $pq$). 

Before the Delaunayhood of $pa$ is restored, we have a co-circularity $p,q,r,a$ in which $q$ leaves $B[p,a,r]\cap \L_{pa}^-$. (This is argued in the proof of Lemma \ref{Lemma:OnceCollin}: Right after the crossing, the point $q$ lies in $B[p,a,r]\cap \L_{pa}^-$, as in Figure \ref{Fig:OrderOrdinaryCrossings} (left), and has to leave that disc before it becomes empty; it cannot cross $pa$ during $J$, when this edge undergoes the Delaunay crossing by $r$). Notice that this is a red-blue co-circularity with respect to $pa$, and a red-red co-circularity with respect to $pq$; see Figure \ref{Fig:OrderOrdinaryCrossings} (right). Since no red-red or blue-blue co-circularities occur during a Delaunay crossing of an edge, the crossing of $pq$ is already over.
\end{proof}

\ignore{
\begin{proof}
Recall that, by our convention, $r$ crosses $pq$ from $\L_{pq}^-$ to $\L_{pq}^+$, and it crosses $pa$ from $\L_{pa}^-$ to $\L_{pa}^+$.
We only prove the part of the lemma concerning the ending times of the clockwise $(p,r)$-crossings, because the proof about the starting times is fully symmetric (by reversing the direction of the time axis). 
By Lemma \ref{Lemma:OnceCollin}, the four points $p,q,a,r$ are involved in at least one red-blue co-circularity with respect to $pa$, which occurs at some time $\zeta\in J$ when $q$ either enters the cap $B[p,a,r]\cap \L_{pa}^+$ or leaves the opposite cap $B[p,a,r]\cap \L_{pa}^-$.
It suffices to prove that the crossing of $pq$ by $r$ is already finished by the above time $\zeta$.

\begin{figure}[htbp]
\begin{center}
\input{OrderCrossingsReturn.pstex_t}\hspace{3cm}\input{OrderCrossings.pstex_t}
\caption{\small Proof of Lemma \ref{Lemma: OrderOrdinaryCrossings}. Left: The co-circularity of $p,q,r,a$ at time $\zeta\in J$ occurs
when $q$ enters the cap $B[p,a,r]\cap \L_{pa}^+$. Then $r$ must have re-entered $\L_{pq}^-$ after $I$ and before $\zeta$.
Right: The co-circularity of $p,q,r,a$ at time $\zeta\in J$ occurs when $q$ leaves the cap $B[p,a,r]\cap \L_{pa}^-$. This is a red-red co-circularity with respect to $pq$, and a red-blue co-circularity with respect to $rq$, so the crossing of $pq$ is already over.}
\label{Fig:OrderCrossings}
\end{center}
\vspace{-0.3cm}
\end{figure}

Assume first that $\zeta$ occurs when $q$ enters the cap $B[p,a,r]\cap \L_{pa}^+$, so
$r$ lies at that moment in $\L_{pq}^-\cap \L_{pa}^-$ (i.e., this is a blue-blue co-circularity with respect to $pq$).
See Figure \ref{Fig:OrderCrossings} (left). Also note that Since $r$ can hit $\L_{pq}$ only once during $I$, $r$ must have re-entered $\L_{pq}^-$ after $I$ and before $\zeta$, and we are done.

Assume, then, that $\zeta$ occurs when $q$ leaves the cap $B[p,a,r]\cap \L_{pa}^-$.
Notice that this is a red-red co-circularity with respect to the edge $pq$, and a red-blue co-circularity with respect to the edge $rq$ whose Delaunayhood is violated right afterwards by $p$ and $a$; see Figure \ref{Fig:OrderCrossings} (right). Since no red-red or blue-blue co-circularities occur during a Delaunay crossing of an edge (or, alternatively, since $rq$ is Delaunay throughout $I$, by Lemma \ref{Lemma:Crossing}), the crossing $(pq,r,I)$ of $pq$ is already over by time $\zeta$.
\end{proof}
}

Lemma \ref{Lemma: OrderOrdinaryCrossings} implies that, for any pair of points $p,r$ in $P$, all the clockwise $(p,r)$-crossings can be linearly ordered by the starting times of their intervals, or by the ending times of their intervals, or by the times when $r$ hits the corresponding edges that emanate from $p$, and all three orders are indentical. Clearly, a symmetric order exists for counterlockwise $(p,r)$-crossings too.


The following theorem provides the long-awaited recursive bound on the maximum number of Delaunay crossings.

\begin{theorem}\label{Thm:OrdinaryCrossings}
Let $12<k_1<k_2$ be a pair of constants. Then the maximum possible number $C_1(n)$ of single Delaunay crossings in any set $P$ of $n$ points, whose pseudo-algebraic motions in $\reals^2$ respects the above assumptions, satisfies the following recurrence:
\begin{equation}\label{Eq:OrdinaryCrossing}
C_1(n)=O\left(k_1^2 N(n/k_1)+k_1 k_2^2 N(n/k_2)+k_1k_2n^2\beta(n)\right),
\end{equation}
where the constant of proportionality in $O(\cdot)$ is independent of $k_1,k_2$.
\end{theorem}
\begin{proof}
Fix a single Delaunay crossing $(pq,r,I=[t_0,t_1])$ as above. If this is the last clockwise $(p,r)$-crossing in the order implied by Lemma \ref{Lemma: OrderOrdinaryCrossings}, then we can charge $(pq,r,I=[t_0,t_1])$ to the edge $pr$. Clearly, this accounts for at most quadratically many single crossings. 

Otherwise, let $(pa,r,J=[t_2,t_3])$ be the clockwise $(p,r)$-crossing that follows immediately after $(pq,r,I=[t_0,t_1])$. That is, we have $t_0<t_2$ and $t_1<t_3$, and no clockwise $(p,r)$-crossings begin in the interval $(t_0,t_2)$ or end in the symmetric interval $(t_1,t_3)$. Refer to Figure \ref{Fig:ChooseNext}.
Note that $(pa,r,J)$ is uniquely determined by the choice of $(pq,r,I)$, and vice versa.
We thus have reduced our problem to bounding the maximum possible number of such ``consecutive" pairs $(pq,r,I=[t_0,t_1])$,  $(pa,r,J=[t_2,t_3])$ of clockwise $(p,r)$-crossings (over all $p,r\in P$).


\begin{figure}[htbp]
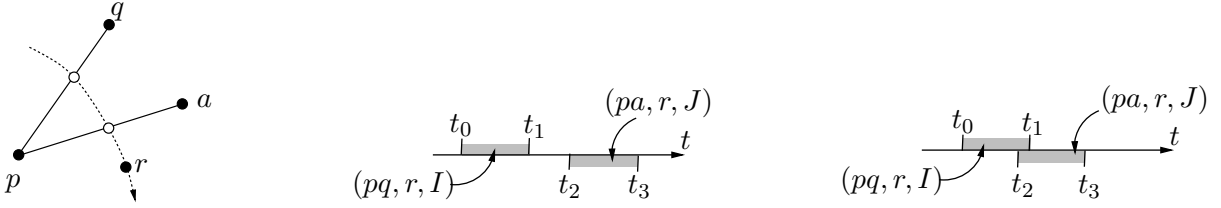

\begin{center}
\input{ConsecutiveCrossings.pstex_t}\hspace{2cm}\input{ConsecutiveDisjoint.pstex_t}\hspace{2cm}\input{ConsecutiveOverlap.pstex_t}
\caption{\small The pair $(pq,r,I=[t_0,t_1]),(pa,r,J=[t_2,t_3])$ of consecutive clockwise $(p,r)$-crossings.
Left: $r$ crosses the edges $pq$ (during $I$, from $\L_{pq}^-$ to $\L_{pq}^+$) and $pa$ (during $J$, from $\L_{pa}^-$ to $\L_{pa}^+$), in this order. Center and right: We have $t_0<t_1$ and $t_1<t_3$, so the intervals $I=[t_0,t_1]$ and $J=[t_2,t_3]$ are either disjoint or {\it partly} overlapping.
No clockwise $(p,r)$-crossings begin in $(t_0,t_2)$ or end in $(t_1,t_3)$.}
\label{Fig:ChooseNext}
\end{center}
\vspace{-0.3cm}
\end{figure}



\medskip
\noindent{\bf Charging events in $\A_{pr}$.} By Lemma \ref{Lemma:Crossing}, $pr$ is Delaunay in each of the intervals $I=[t_0,t_1]$ and $J=[t_2,t_3]$.
If these intervals overlap, then $pr$ is Delaunay throughout $[t_0,t_3]$. Otherwise, as a preparation to the main analysis, we consider the red-blue arrangement $\A_{pr}$ associated with the edge $pr$ during the gap $(t_1,t_2)$ between $I$ and $J$. 
Since $pr$ is Delaunay at both times $t_1$ and $t_2$, we can apply Theorem \ref{Thm:RedBlue} over $(t_1,t_2)$, with the first threshold value $k_1$. Refer to Figure \ref{Fig:ApplyRBPr}.

\begin{figure}[htbp]
\begin{center}
\input{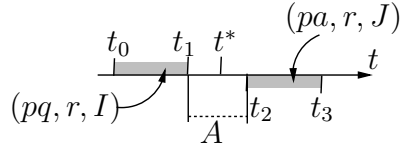}
\caption{\small Applying Theorem \ref{Thm:RedBlue} in $\A_{pr}$ over the gap $(t_1,t_2)$ between $I=[t_0,t_1]$ and $J=[t_2,t_3]$. Note that $pr$ is Delaunay throughout each of the intervals $I,J$. Unless we manage to charge the pair $(pq,r,I),(pa,r,J)$ within $\A_{pr}$, we end up with a subset $A$ of at most $3k_1$ points whose removal extends Delaunayhood of $pr$ to $(t_1,t_2)$.}
\label{Fig:ApplyRBPr}
\end{center}
\vspace{-0.3cm}
\end{figure}

In cases (i) and (ii) of Theorem \ref{Thm:RedBlue}, we charge the pair $(pq,r,I),(pa,r,J)$ either to $\Omega(k_1^2)$ $k_1$-shallow co-circularities, or to a $k_1$-shallow collinearity (within the arrangement of $pr$). 

In both chargings, each shallow co-circularity or collinearity is charged at most a constant number of times. Indeed, consider the moment $t^*$ when the charged event occurs, and notice that it involves $p$ and $r$ (together with one or two additional points of $P$). The choice of $t^*\in (t_1,t_2)\subset (t_1,t_3)$ ensures that the moment $t_1$ (when the crossing of $pq$ by $r$ ends) is the last time before $t^*$ when a clockwise $(p,r)$-crossing is completed. Hence, having guessed $p$ and $r$ (in $O(1)$ ways), $q$ is uniquely determined.
Therefore, using the upper bounds on the number of $k_1$-shallow collinearities and co-circularities established in Section \ref{Sec:Prelim}, we get that the overall number of such consecutive pairs $(pq,r,I=[t_0,t_1]), (pa,r,J=[t_2,t_3])$, for which the red-blue arrangement of $pr$ (during $(t_1,t_2)$) satisfies condition (i) or (ii) of Theorem \ref{Thm:RedBlue}, is 

$$
O\left(k_1n^2\beta(n)+\frac{k_1^4N(n/k_1)}{k_1^2}\right)=O\left(k_1n^2\beta(n)+k_1^2N(n/k_1)\right).
$$

To conclude, we can assume in what follows that either the intervals $I$ and $J$ overlap, or condition (iii) of Theorem \ref{Thm:RedBlue} holds. In the latter case, there exists a subset $A$ of at most $3k_1$ points (possibly including $q$ and/or $a$) so that $pr$ belongs to $\DT(P\setminus A)$ throughout the interval $[t_1,t_2]$. As a matter of fact, $pr$ then belongs to $\DT(P\setminus A)$ throughout an even larger interval $[t_0,t_3]=I\cup (t_1,t_2)\cup J$. 

Notice that reversing the direction of the time axis simply switches the order of $(pq,r,I)$ and $(pa,r,J)$, so their respective points $q$ and $a$ will play symmetrical roles in our case analysis. Recall also that $(pq,r,I)$ is a counterclockwise $(q,r)$-crossing, and $(pa,r,J)$ is a counterclockwise $(a,r)$-crossing (this in addition to their being clockwise $(p,r)$-crossings).

\paragraph{The subsequent chargings--overview.} The rest of the proof is organized as follows.
We distinguish between three possible cases (a)--(c), ruling them out one by one.

In case (a) we assume that $pr$ is hit by one of $q,a$ in the gap $(t_1,t_2)$ between $I$ and $J$, so the respective triple $p,q,r$ or $p,a,r$ performs two single Delaunay crossings in, respectively, $(P\setminus A)\cup \{q\}$ or $(P\setminus A)\cup \{a\}$. Hence, our analysis bottoms out via Lemma \ref{Lemma:TwiceCollin}.

In case (b) we assume that the edge $rq$ is never Delaunay in the interval $[t_3,\infty)$, or that the edge $ra$ is never Delaunay in the symmetric interval $(-\infty,t_0]$.
In the first sub-scenario, we show that $(pq,r,I=[t_0,t_1])$ is among the last $3k_1+1$ counterclockwise $(q,r)$-crossings (with respect to the order implied by Lemma \ref{Lemma: OrderOrdinaryCrossings}). Notice that, by Lemma \ref{Lemma:Crossing}, no such crossings can begin or end after $t_3$ (where $rq$ is not even Delaunay), so we are only to show that at most $3k_1$ $(q,r)$-crossings $(p'q,r,I')$ end after $I$ and before $t_3$ (which is done in Proposition \ref{Prop:Balanced}).
In the second sub-scenario, a fully symmetric argument implies that $(pa,r,J=[t_2,t_3])$ is among the first $3k_1+1$ counterclockwise $(a,r)$-crossings. In both sub-scenarios, the overall number of such consecutive pairs $(pq,r,I),(pa,r,J)$ is easily seen to be $O(k_1n^2)$.
 
Finally, in case (c) we may assume that there exists a time $t_{rq}\geq t_3$ which is the {\it first} such time when $rq$ belongs to $\DT(P)$, and that
there exists a time $t_{ra}\leq t_0$ which is the {\it last} such time when $ra$ belongs to $\DT(P)$.
We argue that the edge $rq$ is hit in $(t_1,t_{rq}]$ by one of $p,a$, and that the edge $ra$ is hit in the symmetric interval $[t_{rq},t_2)$ by one of $p,q$.
We then invoke Theorem \ref{Thm:RedBlue} and try to charge $(pq,r,I), (pa,r,J)$ within one of the red-blue arrangements $\A_{ra},\A_{rq}$. In the case of failure, each of the above additional crossings of $rq$ and $ra$ yields a Delaunay crossing with respect a suitably reduced subset of $P$, so at least one of the triples $\{p,q,r\},\{p,a,r\},\{q,r,a\}$ is involved in two Delaunay crossings. Hence, our analysis again bottoms out via Lemma \ref{Lemma:TwiceCollin}.




\paragraph{Case (a).} The above intervals $I=[t_0,t_1]$ and $J=[t_2,t_3]$ are disjoint and (at least) one of the points $q,a$ hits the edge $pr$ in the interval $(t_1,t_2)$. 
(By Lemma \ref{Lemma:Crossing}, none of $q,a$ can hit $pr$ in $I$ or $J$.) Refer to Figure \ref{Fig:CrossPr}.

Assume, with no loss of generality, that $pr$ is hit in $(t_1,t_2)$ by $q$.
Since $pr$ is Delaunay at both times $t_1$ and $t_2$, the edge $pr$ (or its reversely oriented copy $rp$) undergoes a Delaunay crossing by $q$ within the smaller triangulation $\DT((P\setminus A)\cup\{q\})$ during some sub-interval of $(t_1,t_2)$. 
This is in addition to the inherited single Delaunay crossing of $pq$ by $r$, which is easily checked to occur in $\DT((P\setminus A)\cup\{q\})$ too.
Recalling the assumptions on the possible collinearities in $P$, we get that both of these crossings in $\DT((P\setminus A)\cup\{q\})$ must be single Delaunay crossings.
Lemma \ref{Lemma:TwiceCollin}, combined with the probabilistic argument of Clarkson and Shor \cite{CS}, in a manner similar to that used in Section \ref{Sec:Prelim}, provides an upper bound of $O(k_1n^2)$ on the number of such triples $p,q,r$. Clearly, this also bounds the overall number of such consecutive pairs $(pq,r,I),(pa,r,J)$. 

\begin{figure}[htbp]
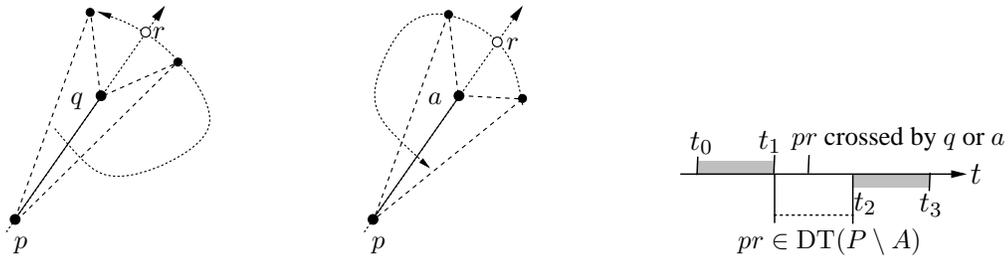

\begin{center}
\input{SecondCrossing.pstex_t}\hspace{2cm}\input{SecondCrossingReverse.pstex_t}\hspace{2cm}\input{SchemeA.pstex_t}
\caption{\small Case (a). Left and center: The edge $pr$ is hit by (at least) one of $q,a$ during $(t_1,t_3)$. Right: The edge $pr$ undergoes a Delaunay crossings by $q$ or $a$ within an appropriate triangulation $\DT((P\setminus A)\cup \{q\})$ or $\DT((P\setminus A)\cup \{a\})$.}
\label{Fig:CrossPr}
\end{center}
\end{figure}

Symmetrically, if $pr$ is hit in the interval $(t_1,t_2)$ by $a$, the triple $p,a,r$ performs two single Delaunay crossings in the triangulation $\DT((P\setminus A)\cup \{a\})$. By Lemma \ref{Lemma:TwiceCollin}, and again using the probabilistic argument of Clarkson and Shor, the overall number of such crossings $(pa,r,J)$ (and, hence, of such consecutive pairs $(pq,r,I),(pa,r,J)$) too cannot exceed $O(k_1n^2)$.

\bigskip
To conclude, in each of the subsequent cases (b)--(c) we may assume that the preceding scenario (a) does not occur. In addition, we continue to assume that, unless the intervals $I=[t_0,t_1]$ and $J=[t_2,t_3]$ overlap, there is a subset $A$ of at most $3k_1$ points whose removal restores the Delaunayhood of $pr$ in the gap $(t_1,t_2)$ between $I$ and $J$.

\begin{proposition}\label{Prop:Balanced}
With the above assumptions, at most $3k_1$ counterclockwise $(q,r)$-crossings $(p'q,r,I')$ end in the interval $(t_1,t_3)$, and at most $3k_1$ counterclockwise $(a,r)$-crossings $(p'a,r,J')$ begin in the symmetric interval $(t_0,t_2)$.
\end{proposition}
\begin{proof}
With no loss of generality, we focus on counterclockwise $(q,r)$-crossings. The counterclockwise $(a,r)$-crossings are handled symmetrically, by reversing the direction of the time axis. 

Let $(p'q,r,I')$ be a counterclockwise $(q,r)$-crossing that ends in $(t_1,t_3)$. In particular, $I'$ ends after $I$ so, by the counterclockwise variant of Lemma \ref{Lemma: OrderOrdinaryCrossings}, $I'$ also begins after the beginning $t_0$ of $I$. Therefore, we get that $I'\subset (t_0,t_3]$.
We claim that the intervals $I=[t_0,t_1]$ and $J=[t_2,t_3]$ are disjoint, and the respective point $p'$ of $(p'q,r,I')$ belongs to the above set $A$ of at most $3k_1$ points whose removal restores the Delaunayhood of $pr$ throughout $[t_1,t_2]$. This will imply that the overall number of such crossings $(p'q,r,I')$ cannot exceed $3k_1$. 

\begin{figure}[htbp]
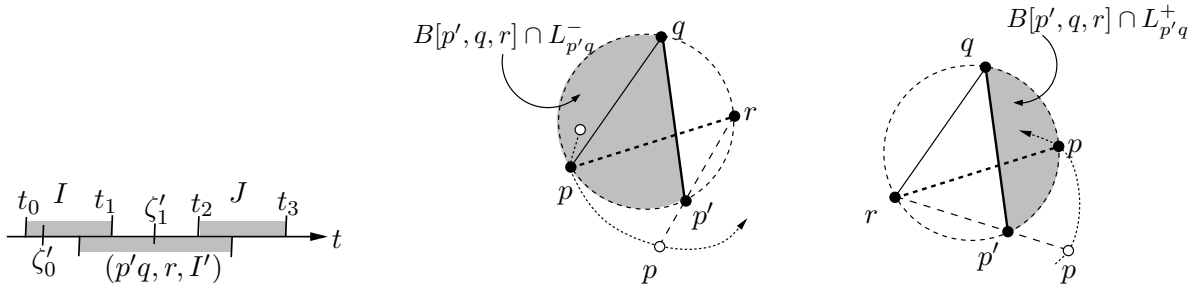

\begin{center}
\input{SchemeBalanced.pstex_t}\hspace{1cm}\input{BalancedCross.pstex_t}\hspace{1.5cm}\input{BalancedCross1.pstex_t}
\caption{\small Proof of Proposition \ref{Prop:Balanced}. Left: The crossing $(p'q,r,I')$ occurs within $(t_0,t_3]$. The points $p,q,r,p'$ are co-circular at times $\zeta'_0\in I\setminus I'$ and $\zeta'_1\in I'\setminus I$. The latter co-circularity (at time $\zeta'_1$) is red-blue with respect to $pr$, so it occurs in the gap $(t_1,t_2)$ between $I$ and $J$. Center: A possible trajectory of $p$ during $(\zeta'_1,t_2)$ if $r$ lies in $\L_{p'q}^+$ at time $\zeta'_1$. Right: A possible trajectory of $p$ during $(t_1,\zeta'_1)$ if $r$ lies in $\L_{p'q}^-$ at time $\zeta'_1$.}
\label{Fig:Balanced}
\end{center}
\end{figure} 

Indeed, by Lemma \ref{Lemma:OnceCollin}, the four points $p,q,r,p'$ are involved in (at least) one co-circularity during the single Delaunay crossing of $pq$ by $r$, and in another co-circularity during the similar crossing of $p'q$ by $r$. Refer to Figure \ref{Fig:Balanced} (left).
Specifically, the former co-circularity is red-blue with respect to the two diagonal edges $pq$ and $p'r$. By Lemma \ref{Lemma:Crossing}, $p'r$ is Delaunay throughout $I'$, so this co-circularity occurs at some time $\zeta'_0\in I\setminus I'$. Similarly, the other co-circularity of $p,q,r,p'$ is red-blue with respect to the edges $p'q$ and $pr$, so it occurs at some later time $\zeta'_1\in I'\setminus I$. 
Since the latter co-circularity, occurring at time $\zeta'_1$, is red-blue with respect to $pr$, it cannot occur during during the interval $J$ (where $pr$ is Delaunay). Hence, $\zeta'_1$ must occur in the gap $(t_1,t_2)$ between the intervals $I$ and $J$, which then cannot overlap. 

We next argue that $p'$ hits $pr$ in the above gap interval $(t_1,t_2)$, which will immediately\footnote{Clearly, the Delaunayhood of $pr$ in $(t_1,t_1)$ cannot be restored before we remove from $P$ every point that crosses $pr$ in that interval.} imply that $p'\in A$.
Since $\zeta'_0\leq t_1<\zeta'_1$, the times $\zeta'_0$ and $\zeta'_1$ cannot coincide, so $\zeta'_1$ is the {\it only} co-circularity of $p,q,r,p'$ in $(t_1,t_2)$. To obtain the asserted crossing of $pr$, we distinguish between the following two sub-cases:

\medskip
\noindent (i) Assume first that $r$ lies in $\L_{p'q}^+$ at time $\zeta'_1$. As prescribed in Lemma \ref{Lemma:OnceCollin}, this co-circularity occurs when $p$ leaves the cap $B[p',q,r]\cap \L_{p'q}^-$, so the Delaunayhood of $pr$ is violated right after time $\zeta'_1$ by $q$ and $p'$. See Figure \ref{Fig:Balanced} (center).
By Lemma \ref{Lemma:MustCross}, and keeping in mind that $pr$ is Delaunay at time $t_2$ (and no further co-circularities of $p,q,r,p'$ can occur in $(t_1,t_2)$), the edge $pr$ is hit by at least one of the two points $q,p'$ at some moment in $(\zeta'_1,t_2)$. 
Since case (a) is excluded, $q$ cannot hit $pr$ during that interval (which is contained in $(t_1,t_2)$), so it must be the case that $p'$ hits $pr$ during the time interval $(t_1,t_2)$. 

\smallskip
\noindent (ii) Assume, then, that $r$ lies in $\L_{p'q}^-$ at time $\zeta'_1$. As prescribed in Lemma \ref{Lemma:OnceCollin}, this co-circularity occurs when $p$ enters the cap $B[p',q,r]\cap \L_{p'q}^+$, so the Delaunayhood of $pr$ is violated right {\it before} time $\zeta'_1$ by $p'$ and $q$. See Figure \ref{Fig:Balanced} (right).
Since $pr$ is Delaunay at time $t_1<\zeta'_1$ (and no further co-circularities of $p,q,r,r'$ can take place in $(t_1,\zeta'_1)$), we can apply Lemma \ref{Lemma:MustCross} from $\zeta'_1$ in the reverse direction of the time axis to get that $pr$ is hit by at least one of $q,p'$ at some moment in $(t_1,\zeta'_1)$. Since case (a) is excluded, it must be the case that $p'$ hits $pr$ during the time interval $(t_1,t_2)$.

\smallskip
To conclude, we have shown that $pr$ is hit by $p'$ in the gap interval $(t_1,t_2)$ between $I$ and $J$.
Therefore, $p'$ belongs to the above set $A$ of at most $3k_1$ points whose removal restores the Delaunayhood of $pr$ throughout the inerval $[t_1,t_2]$, so the overall number of such counterclockwise $(q,r)$-crossings $(p'q,r,I')$ cannot exceed $3k_1$.

Repeating the above analysis in the reverse direction of the time axis shows that, if $(p'a,r,J')$ is a counterclockwise $(a,r)$-crossing starting in $(t_0,t_2)$, its respective point $p'$ crosses $pr$ in the gap $(t_1,t_2)$ and, therefore, again belongs to $A$.
Hence, the overall number of such crossings $(p'a,r,J')$ is at most $3k_1$ too.
\end{proof}




\smallskip
\noindent {\bf Case (b).} The edge $rq$ is never Delaunay in the interval $[t_3,\infty)$, or the edge $ra$ is never Delaunay in the symmetric interval $(-\infty,t_0]$.

If $rq$ is never Delaunay in $[t_3,\infty)$, then no counterclockwise $(q,r)$-crossings $(p'q,r,I')$ can occur (i.e., begin or end) after $t_3$, for, by Lemma \ref{Lemma:Crossing}, $rq$ must belong to $\DT(P)$ when any such Delaunay crossing takes place. 
Combining this with Proposition \ref{Prop:Balanced}, we conclude that $(pq,r,I)$ is among the $3k_1+1$ counterclockwise $(q,r)$-crossings $(p'q,r,I')$ that end the latest. In other words, $(pq,r,I)$ is among the last $3k_1+1$ counterclockwise $(q,r)$-crossings with respect to the order implied by Lemma \ref{Lemma: OrderOrdinaryCrossings}. Clearly, this scenario can happen for at most $O(k_1n^2)$ consecutive pairs $(pq,r,I), (pa,r,J)$ of Delaunay crossings. 

A fully symmetric argument applies if $ra$ is never Delaunay in $(-\infty,t_0]$.
In that case, we get that $(pa,r,J)$ is among the first $3k_1+1$ counterclockwise $(a,r)$-crossings $(p'a,r,J')$, which can happen for at most $O(k_1n^2)$ pairs $(pq,r,I),(pa,r,J)$.

\medskip
\noindent {\bf Case (c).} Neither of the previous two cases holds. 
In particular, there exists a time $t_{rq}\geq t_3$ which is the {\it first} such time when $rq$ belongs to $\DT(P)$.
Similarly, there exists a time $t_{ra}\leq t_0$ which is the {\it last} such time when $ra$ belongs to $\DT(P)$. See Figure \ref{Fig:SchemeCrossRq}.

More precisely, if $rq$ is Delaunay at time $t_3$, then we have $t_{rq}=t_3$. Otherwise, if $t_{rq}>t_3$, this is one of the critical times when $rq$ enters $\DT(P)$. As reviewed in Section \ref{Sec:Prelim}, $\DT(P)$ experiences then either a Delaunay co-circularity or a hull event (where $rq$ is hit by some point of $P\setminus \{r,q\}$). In the latter case, $rq$ is not strictly Delaunay at time $t_{rq}$ and appears in $\DT(P)$ only {\it right after} this event. The time $t_{ra}$ has fully symmetrical properties. For simplicity of presentation, we consider the edges $rq$ and $ra$ to be Delaunay at the respective times $t_{rq}$ and $t_{ra}$.

\begin{figure}[htbp]
\begin{center}
\input{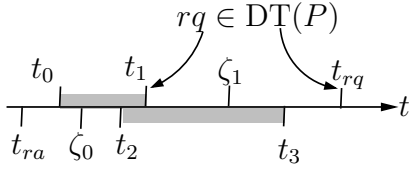}
\caption{\small Proposition \ref{Prop:ExtraCollin}. The edge $rq$ is hit in $(t_1,t_{rq}]$ by at least one of $p,a$. Symmetrically, the edge $ra$ is hit in $[t_{ra},t_2)$ by at least one of $p,q$.
The four points $p,q,r,a$ are co-circular at times $\zeta_0\in I\setminus J$ and $\zeta_1\in J\setminus I$.
Note that $rq$ is Delaunay at both times $t_1,t_{rq}$, and $ra$ is Delaunay at both times $t_{ra},t_2$.}
\label{Fig:SchemeCrossRq}
\end{center}
\vspace{-0.6cm}
\end{figure}

\begin{proposition}\label{Prop:ExtraCollin}
With the above assumptions, the edge $rq$ is hit in the interval $(t_1,t_{rq}]$ by at least one of the points $a,p$, and, symmetrically, the edge $ra$ is hit in the interval $[t_{ra},t_2)$ by at least one of the points $p,q$.
\end{proposition}
\begin{proof}
Clearly, it is sufficient to show $rq$ is hit in $(t_1,t_{rq}]$ by $a$ and/or $p$. The symmetric crossing of $ra$ by $a$ and/or $p$ is then obtained by repeating the same analysis in the time-reversed frame (thereby switching the roles of $q$ and $a$).

Indeed, applying Lemma \ref{Lemma:OnceCollin} for the two crossings $(pq,r,I)$ and $(pa,r,J)$ shows that the four points $p,q,a,r$ are co-circular in each of the intervals $I=[t_0,t_1]$ and $J=[t_2,t_3]$. Specifically, the former co-circularity (in $I$) is red-blue with respect to the edges $pq$ and $ra$, so it occurs at some time $\zeta_0\in I\setminus J$ (because $ra$ is Delaunay throughout $J$). The latter co-circularity is red-blue with respect to $pa$ and $rq$, so it must occur at some time $\zeta_1$ in the symmetric interval $J\setminus I$. See Figure \ref{Fig:SchemeCrossRq} for a schematic summary.

Clearly, the above two co-circularities of $p,q,a,r$ cannot coincide, so $\zeta_1$ is the {\it only} co-circularity of $p,q,r,a$ in the interval $(t_1,t_{rq}]$ (which contains $J\setminus I$). 
To obtain the asserted crossing of $rq$ by $p$ or/and $a$, we distinguish between the following two sub-cases.

\begin{figure}[htbp]
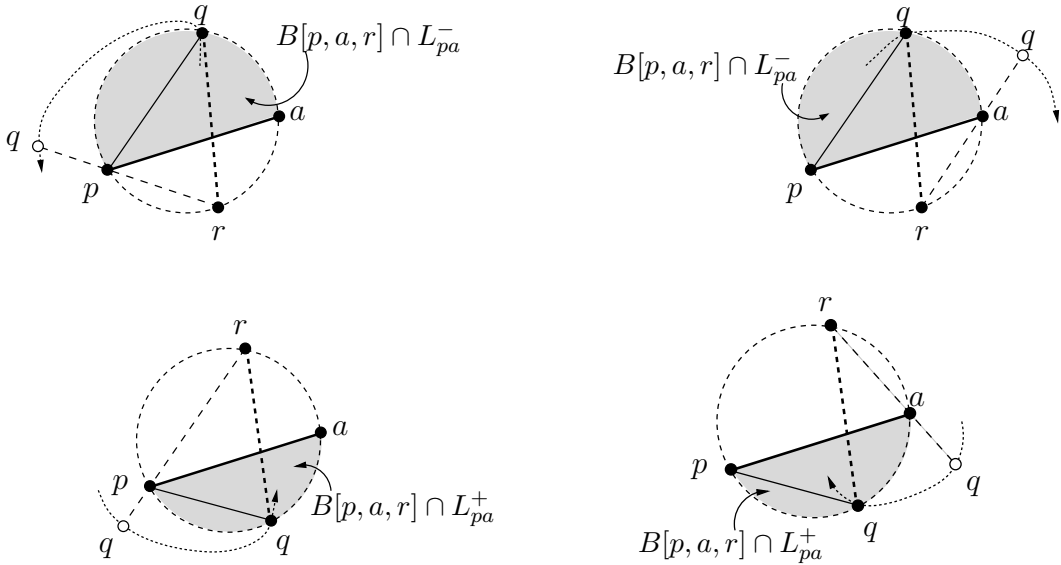

\begin{center}
\input{CrossQr.pstex_t}\hspace{4cm}\input{CrossQr1.pstex_t}\\
\vspace{0.5cm}
\input{CrossQrBefore.pstex_t}\hspace{4cm}\input{CrossQrBefore1.pstex_t}
\caption{\small Proof of Proposition \ref{Prop:ExtraCollin}. Arguing that $rq$ is hit by (at least) one of $p,a$ in $(t_1,t_{rq}]$. Top: Possible trajectories of $q$ if it leaves the cap $B[p,a,r]\cap \L_{pa}^-$ at time $\zeta_1$. The asserted crossing of $rq$ occurs in $(\zeta_1,t_{rq}]$. Bottom: Possible trajectories of $q$ if it enters the opposite cap $B[p,a,r]\cap \L_{pa}^+$ at time $\zeta_1$, so the asserted crossing occurs in $(t_1,\zeta_1)$.}
\label{Fig:ExtraCollin}
\end{center}
\end{figure} 

\smallskip
\noindent(i) The point $r$ lies at time $\zeta_1$ in $\L_{pa}^+$.
As prescribed in Lemma \ref{Lemma:OnceCollin}, this co-circularity occurs when $q$ leaves the cap $B[p,a,r]\cap \L_{pa}^-$, so the Delaunayhood of $rq$ is violated right after time $\zeta_1$ by $p$ and $a$.
By Lemma \ref{Lemma:MustCross}, and keeping in mind that $rq$ is Delaunay at time $t_{rq}$, the edge $rq$ is hit by at least one of the two points $a,p$ at some moment in $(\zeta_1,t_{rq}]$. See Figure \ref{Fig:ExtraCollin} (top).

\smallskip
\noindent(ii) The point $r$ lies at time $\zeta_1$ in $\L_{pa}^-$.
As prescribed in Lemma \ref{Lemma:OnceCollin}, this co-circularity occurs when $q$ enters the cap $B[p,a,r]\cap \L_{pa}^+$, so the Delaunayhood of $rq$ is violated right {\it before} time $\zeta_1$ by $p$ and $a$.
By Lemma \ref{Lemma:MustCross}, and keeping in mind that $rq$ is Delaunay at time $t_1$, the edge $rq$ is hit by at least one of the two points $a,p$ at some moment in $(t_1,\zeta_1)$. See Figure \ref{Fig:ExtraCollin} (bottom).
\end{proof}

Combining the new collinearities in Proposition \ref{Prop:ExtraCollin} with the already existing crossings of $pq$ and $pa$ by $r$ shows that at least one of the triples $\{p,q,r\}$, $\{p,a,r\}$, $\{q,r,a\}$ performs two collinearities, of distinct order types.
If we manage to amplify the above additional crossings of $rq$ and $ra$ into full-fledged Delaunay crossings (as we did in Section \ref{Sec:DelCocircs} and in case (a)), then some sub-triple in $p,q,a,r$ will necessarily perform two single Delaunay crossings, so our analysis will bottom out via Lemma \ref{Lemma:TwiceCollin}.

As a preparation,
we first apply Theorem \ref{Thm:RedBlue} in $\A_{rq}$ over the interval $(t_1,t_{rq})$, and then apply it in $\A_{ra}$ over $(t_{ra},t_2)$, both times with the second constant parameter $k_2>k_1$ instead of $k$. (We again emphasize that $rq$ is Delaunay at both endpoints of $(t_1,t_{rq})$, and $ra$ is Delaunay at both endpoints of $(t_{ra},t_2)$.)

\smallskip
\noindent{\bf Charging events in $\A_{rq}$.}
Consider the first application of Theorem \ref{Thm:RedBlue}. Refer to Figure \ref{Fig:ExtendRq}.
If one of the first two conditions of Theorem \ref{Thm:RedBlue} holds, we can charge the pair $(pq,r,I),(pa,r,J)$ within $\A_{rq}$ either to $\Omega(k_2^2)$
$k_2$-shallow co-circularities, or to a $k_2$-shallow collinearity. Here the crucial observation is that every co-circularity or collinearity (which occurs in $(t_1,t_{rq})$ and involves $r$ and $q$) is charged in this manner at most $O(k_1)$ times. Indeed, by Proposition \ref{Prop:Balanced},
at most $3k_1$ counterclockwise $(q,r)$-crossings end in $(t_1,t_3)$. Moreover, unless $t_{rq}=t_3$, no $(q,r)$-crossings can even partly overlap (let alone end in) $[t_3,t_{rq})$, until $rq$ returns to $\DT(P)$ at time $t_{rq}$.
Thus, $(pq,r,I)$ is among the $3k_1+1$ counterclockwise $(q,r)$-crossings that are the latest to end before any of the charged collinearity or co-circularity events (all occurring during $(t_1,t_{rq})$). 
Arguing as in the previous chargings, the number of consecutive pairs $(pq,r,I), (pa,r,J)$ for which such a charging applies is at most $O(k_1k_2^2N(n/k_2)+k_1k_2n^2\beta(n))$.

Finally, if condition (iii) of Theorem \ref{Thm:RedBlue} holds then the Delaunayhood of $rq$ can be restored, throughout the interval $[t_1,t_{rq}]$, by removing a set $B_{rq}$ of at most $3k_2$ points of $P$ (including $p$ and/or $a$). (By Lemma \ref{Lemma:Crossing}, $rq$ is also Delaunay throughout $I=[t_0,t_1]$, so its Delaunayhood extends, in $\DT(P\setminus B_{rq})$, to an even larger interval $I\cup [t_1,t_{rq}]=[t_0,t_{rq}]$.)
Recalling Proposition \ref{Prop:ExtraCollin}, we distinguish between the following two subcases.

\begin{figure}[htbp]
\begin{center}
\input{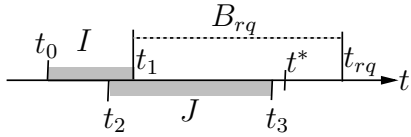}
\caption{\small Applying Theorem \ref{Thm:RedBlue} in $\A_{rq}$ over $(t_1,t_{rq})$--a schematic summary. The edge $rq$ is Delaunay at both times $t_1,t_{rq}$. In cases (i), (ii), each $k_2$-shallow event is charged only $O(k_1)$ times because $(pq,r,I=[t_0,t_1])$ is among the last $3k_1+1$ counterclockwise $(q,r)$-crossings to end before the respective time $t^*$ of the event. In case (iii) we have a subset $B_{rq}$ of at most $3k_2$ points whose removal extends the Delaunayhood of $rq$ to $(t_1,t_{rq})$.}
\label{Fig:ExtendRq}
\end{center}
\end{figure} 

If $rq$ is hit in $(t_1,t_{rq}]$ by $p$, then
the smaller set $(P\setminus B_{rq})\cup\{p\}$ yields a Delaunay crossing of $rq$ (or of its reversely oriented copy $qr$) by $p$. This is in addition to the inherited Delaunay crossing of $pq$ by $r$. As in case (a), it is easy to check that both of these crossings in $\DT((P\setminus B_{rq})\cup\{p\})$ must be single Delaunay crossings. 
Hence, Lemma \ref{Lemma:TwiceCollin}, combined with the Clarkson-Shor argument \cite{CS}, in a manner similar to that used in Section \ref{Sec:Prelim} and the previous cases, provides an upper bound of $O(k_2n^2)$ on the number of such triples $p,q,r$. Clearly, this also bounds the overall number of such consecutive pairs $(pq,r,I),(pa,r,J)$ of Delaunay crossings.

To conclude, we may assume that $rq$ is hit in $(t_1,t_{rq})$ by $a$, so the smaller set $(P\setminus B_{rq})\cup\{a\}$ yields a Delaunay crossing of $rq$ by $a$.

\smallskip
\noindent{\bf Charging events in $\A_{ra}$.} The second application of Theorem \ref{Thm:RedBlue} in $\A_{ra}$ over $(t_{ra},t_2)$ is fully symmetric to the first one. Refer to Figure \ref{Fig:ExtendRa}. If at least one of conditions (i), (ii) is satisfied, we charge the pair $(pq,r,I), (pa,r,J)$ within $\A_{ra}$ either to $\Omega(k_2^2)$ $k_2$-shallow co-circularities or to an $k_2$-shallow collinearity that occur in $\A_{ra}$ during that interval. The crucial observation is that $(pa,r,J)$ is among the first $3k_1+1$ counterclockwise $(a,r)$-crossings to begin after each charged event, which also involves $a$ and $r$. Hence, every collinearity or co-circularity is charged at most $O(k_1)$ times, so this charging accounts for at most $O(k_1k_2^2N(n/k_2)+k_1k_2n^2\beta(n))$ pairs.

For each of the remaining pairs $(pq,r,I),(pa,r,J)$ we have a set $B_{ra}$ of at most $3k_2$ points (possibly including $p$ and/or $q$) whose removal restores the Delaunayhood of $ra$ throughout $[t_{ra},t_2]$. (By Lemma \ref{Lemma:Crossing}, $ra$ is also Delaunay throughout $J=[t_2,t_3]$, so its Delaunayhood extends, in $\DT(P\setminus B_{ra})$, to an even larger interval $[t_{ra},t_2]\cup J=[t_{ra},t_{3}]$.)
To complete the proof of Theorem \ref{Thm:OrdinaryCrossings}, we again recall Proposition \ref{Prop:ExtraCollin} and distinguish between the two possible crossings of $ra$.

\begin{figure}[htbp]
\begin{center}
\input{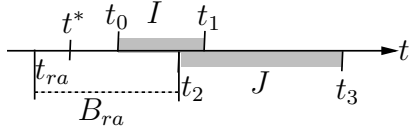}
\caption{\small Applying Theorem \ref{Thm:RedBlue} in $\A_{ra}$ over $(t_{ra},t_2)$--a schematic summary. The edge $ra$ is Delaunay at both times $t_{ra},t_2$. In cases (i), (ii), each $k_2$-shallow event is charged only $O(k_1)$ times because $(pa,r,J=[t_2,t_3])$ is among the first $3k_1+1$ counterclockwise $(a,r)$-crossings to begin after the respective time $t^*$ of the event. In case (iii) we have a subset $B_{ra}$ of at most $3k_2$ points whose removal extends the Delaunayhood of $ra$ to $(t_{ra},t_2)$.}
\label{Fig:ExtendRa}
\end{center}
\end{figure}

If $ra$ is hit in $[t_{ra},t_2)$ by $p$, then
the smaller set $(P\setminus B_{ra})\cup\{p\}$ yields a Delaunay crossing of $ra$ (or of its reversely oriented copy $qr$) by $p$, and a Delaunay crossing of $pa$ by $r$, which are easily checked to be single Delaunay crossings. 
Hence, Lemma \ref{Lemma:TwiceCollin}, combined with the Clarkson-Shor argument \cite{CS}, provides an upper bound of $O(k_2n^2)$ on the number of such triples $p,q,a$, which also bounds the overall number of such consecutive pairs $(pq,r,I),(pa,r,J)$.

Finally, if $ra$ is hit in $[t_{ra},t_{2})$ by $q$, the triple $a,r,q$ performs two Delaunay crossings in the triangulation $\DT((P\setminus (B_{rq}\cup B_{ra}))\cup \{q,a\})$, that is, the crossing of $rq$ by $a$ (occurring entirely within $(t_1,t_{rq}]$), and the crossing of $ra$ by $q$ (occurring entirely within $[t_{ra},t_2)$). The standard assumptions on the possible collinearities in $P$ readily imply that both of these crossings are in fact single Delaunay crossings.
Combining Lemma \ref{Lemma:TwiceCollin} with the probabilistic argument of Clarkson and Shor \cite{CS}, as above, we get that the number of such triples $q,r,a$ is at most $O(k_2n^2)$. Note, though, that our goal is to bound the number of possible pairs $p,q,r$ (or, alternatively, $p,a,r$) rather than $q,r,a$. However, recall that $a$ hits $rq$ during the time interval $(t_1,t_{rq}]$, and $(pq,r,I)$ is then among the $3k_1+1$ counterclockwise $(q,r)$-crossings that end latest before that collinearity of $q,r,a$. Hence, any triple $q,r,a$ can arise in the charging for at most $O(k_1)$ triples $p,q,r$. In conclusion, the number of consecutive pairs $(pq,r,I), (pa,r,J)$ that fall into this final subcase is at most $O(k_1k_2n^2)$. 

Adding up the bounds obtained in cases (a)--(c),  and in the preparatory charging of $k_1$-shallow events in $\A_{pr}$, the theorem follows.
\end{proof}

\subsection{The number of double Delaunay crossings}\label{Subsec:Double}
In this subsection we show that any set $P$ of $n$ points moving as above in $\reals^2$ admits at most $O(n^2)$ double Delaunay crossings.
Since double Delaunay crossings are not possible if no ordered triple of points can be collinear more than once
(i.e., if for any $p,q,r$ the third point $r$ can hit the segment $pq$ at most once), we may assume throughout this subsection that no triple of points in $P$ can be collinear more than twice.

Without loss of generality, we only bound the number of such double Delaunay crossings $(pq,r,I)$ whose point $r$ crosses through $pq$ from $\L_{pq}^-$ to $\L_{pq}^+$ during the first collinearity of $p,q,r$ (and then returns back to $\L_{pq}^-$ during the second collinearity).
Indeed, if the crossing $(pq,r,I)$ does not satisfy the above condition then they are satisfied by $(qp,r,I)$.
Our goal is to show that (on average) a point $r$ of $P$ is involved in only few Delaunay crossings of edges that share the same endpoint $p$.

The following theorem provides certain structural properties of two double crossings that share the same crossing point ($r$) and one endpoint ($p$) of the crossed edges.

\begin{figure}[htbp]
\begin{center}
\input{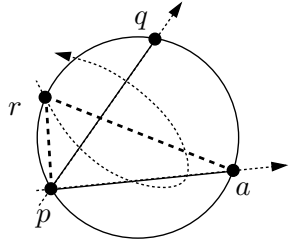}
\caption{\small The trace of $r$ according to Theorem \ref{Thm:OrderSpecialCrossings}. The four points $p,q,a,r$ are involved during $I$ in two co-circularities, which are red-blue with respect to the edges $pq$ and $ra$.}
\label{Fig:StayDelaunay}
\end{center}
\end{figure} 

\begin{theorem}\label{Thm:OrderSpecialCrossings}
Let $(pq,r,I)$ and $(pa,r,J)$ be two double Delaunay crossings of $p$-edges (that is, edges incident to $p$) $pq,pa$ by the same point $r$. 
Assume that
the first collinearity of $p,q,r$ occurs before the first collinearity of $p,a,r$.
Then the following properties hold (with the conventions assumed above):\\
\indent(i) $a$ lies in $\L_{pq}^+$ at both times when $r$ hits $pq$.\\
\indent (ii) $q$ lies in $\L_{pa}^-$ at both times when $r$ hits $pa$.\\
\indent (iii) The points $p,q,a,r$ are involved during $I\setminus J$ in two co-circularities, both of them red-blue with respect to $pq$ and occurring when $r\in \L^-_{pq}$ and $a\in \L^+_{pq}$. \\
\indent (iv) One of the two co-circularities in (iii) occurs before the beginning of $J$; right before it the Delaunayhood of $ra$ is violated by $p$ and $q$. A symmetric such co-circularity occurs after the end of $J$;
right after it the Delaunayhood of $ra$ is again violated by $p$ and $q$.
In particular, $J\subset I$.
\end{theorem}
The schematic description of the motion of $r$ during $I$, according to the above theorem, is depicted in Figure \ref{Fig:StayDelaunay} (right).
Clearly, a suitable variant of Theorem \ref{Thm:OrderSpecialCrossings} exists also for similar pairs of double crossings of incoming $p$-edges $qp,ap$ that are oriented {\it towards} $p$ (again, by the same point $r$).
\begin{proof}
We first establish Part (ii) of the theorem.
The crucial observation is that the first collinearity of $p,a,r$ occurs when $r$ lies in $\L_{pq}^+$ (i.e., during the interval between the two collinearities of $p,q,r$). 
Indeed, otherwise the point $a$ must lie in $\L_{pq}^+=\L_{pr}^+$ at both collinearities of $p,a,r$, and $q$ must lie in $\L_{pa}^+$ at both collinearities of $p,a,r$. 
We shall prove that, in this hypothetical setup, the points $p,q,a,r$ are involved in two co-circularities during $I$ which are red-blue with respect to $pq$, and in a symmetric pair of co-circularities during $J$, both of them red-blue with respect to $pa$. That will clearly contradict the assumption that any four points can be co-circular at most twice.

Indeed, in the above situation the point $a$ lies in the cap $B[p,q,r]\cap \L_{pq}^+$ shortly before the first collinearity of $p,q,r$, and shortly after their second collinearity. 
Since $B[p,q,r]$ contains no points at the beginning of $I$, the point $a$ must have entered this cap before the first collinearity of $p,q,r$. Moreover, $a$ can enter this cap only through the boundary of $B[p,q,r]$, for otherwise it would hit $pq$ during $I$, and no point of $P\setminus\{p,q,r\}$ can hit $pq$ during its Delaunay crossing by $r$. This argument gives us the first of the promised two red-blue co-circularities that $p,q,a,r$ define with respect to $pq$. The second such co-circularity is symmetric to the first one, and occurs when $a$ leaves the cap $B[p,q,r]\cap \L_{pq}^+$ (and after $r$ returns to $\L_{pq}^-$ through $pq$). See Figure \ref{Fig:DoubleFour} (left).
The other pair of co-circularities, both red-blue with respect to $pa$, is obtained by applying a fully symmetric argument to the cap $B[p,a,r]\cap \L_{pa}^+$ and the point $r$. See Figure \ref{Fig:DoubleFour} (center). (For example, we can switch the roles of $q$ and $a$ by reversing the direction of the time axis.)
Finally, all four co-circularities are distinct, because the same co-circularity cannot be red-blue with respect to two edges $pq,pa$ with a common endpoint. 

\begin{figure}[htbp]
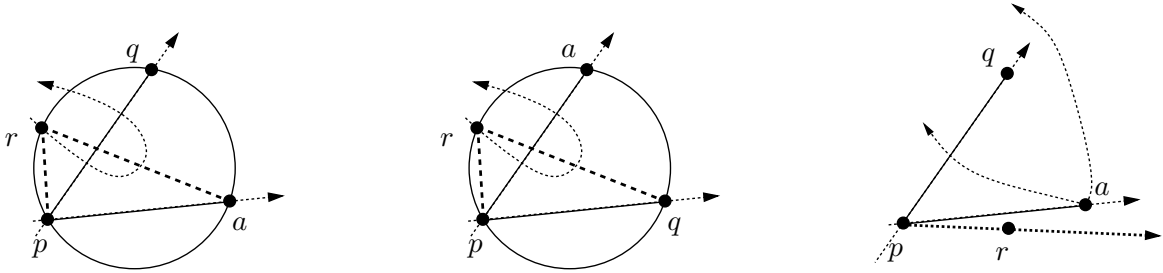

\begin{center}
\input{DoubleFourCocircs1.pstex_t}\hspace{2cm}\input{DoubleFourCocircs.pstex_t}\hspace{2cm}\input{CannotReturn.pstex_t}
\caption{\small Proof of Theorem \ref{Thm:OrderSpecialCrossings}. Left and center: The hypothetical case where $r$ first hits $pa$ within $\L_{pq}^-$, after twice hitting $pq$. The points $p,q,a,r$ are involved in a pair of co-circularities during $I$, and in a symmetric pair of co-circularities during $J$. Right: The hypothetical traces of $a$ if it enters $\L_{pq}^+$ before $r$ (and before the second collinearity of $p,a,r$ occurs).}
\label{Fig:DoubleFour}
\end{center}
\end{figure}

Hence, we can assume, from now on, that the first time when $r$ hits $pa$ occurs when both points lie in $\L_{pq}^+$.
To complete the proof of Part (ii), it suffices to show that the points $a$ and $r$ still remain in $\L_{pq}^+$ during the second collinearity of the triple $p,a,r$.
Indeed, otherwise $a$ must lie in $\L_{pq}^-$ when $r$ hits $pq$ for the second time, because, untill it crosses $pa$ again, $a$ lies in $\L_{pr}^-$ which coincides with $\L_{pq}^-$ at the second crossing of $pq$ by $r$. See Figure \ref{Fig:DoubleFour} (right).
That is, $a$ must cross $\L_{pq}$ from $\L_{pq}^+$ to $\L_{pq}^-$ while $r$ still remains in $\L_{pq}^+$, and before $r$ hits the edges $pq,pa$ for the second time. In particular, the above collinearity of $p,q,a$ must occur during $I\cap J$. Clearly, the point $a$ can potentially cross $\L_{pq}$ in three ways.
If $a$ crosses $\L_{pq}$ within $pq$, this contradicts the definition of $I$ as the interval of the Delaunay crossing of $pq$ by $r$. If $a$ hits $\L_{pq}\setminus pq$ within the ray emanating from $q$ then (at that very moment) $q$ hits $pa$, which contradicts the definition of $J$. Finally, $a$ cannot hit $\L_{pq}\setminus pq$ within the outer ray emanating from $p$ before an additional (and forbidden) collinearity of $p,a,r$ takes place.  This establishes part (ii), and the analysis given above immediately implies part (i) two.

Part (i) follows immediately from Part (ii), because $a$ lies in $\L_{pr}^+$ during both collinearities of $p,q,r$.

Parts (iii) and (iv) follow from Parts (i) and (ii). Indeed, recall that the open disc $B[p,q,r]$ contains no points of $P$ at the beginning of $I$. Right before $r$ hits $pq$ for the first time, the right cap $B[p,q,r]\cap \L_{pq}^+$ of this disc contains $a$. Clearly, $a$ first enters this cap through the corresponding portion of $\partial B[p,q,r]$. This determines the first red-blue co-circularity with respect to $pq$, right before which the Delaunayhood of $ra$ is violated by $p$ and $q$. The symmetric such co-circularity occurs during $I$ when the point $a$ leaves the cap $B[p,q,r]\cap \L_{pq}^+$, after the second collinearity of $p,q,r$. Clearly, the Delaunayhood of $ra$ is violated right after that co-circularity by $p$ and $q$. By Lemma \ref{Lemma:Crossing}, neither of these co-circularities can occur during $J$, because $ra$ remains Delaunay throughout $J$. Hence, the former one occurs, according to the previously established Parts (i) and (ii), before $J$, and the latter one occurs after $J$. This establishes parts (iii) and (iv), and completes the proof.
\end{proof}

\begin{theorem}\label{Thm:SpecialCrossings}
Let $P$ be a set of $n$ points, whose motion in $\reals^2$ respects the following conventions: (i) any four points can be co-circular at most twice, and (ii) no three points can be collinear more than twice. Then $P$ admits at most $O(n^2)$ double Delaunay crossings.
\end{theorem}
\begin{proof}
We fix a pair of points $p,r$ in $P$. Our strategy is to show that, for an average such pair, there is at most a constant number of double Delaunay crossings of $p$-edges by $r$.
Indeed, let $(pq_1,r,I_1),(pq_2,r,I_2),$ $\ldots,(pq_k,r,I_k)$ be the complete list of such double Delaunay crossings of $p$-edges by $r$, and assume that $r$ hits the edges $pq_1,pq_2,\ldots,p_{q_k}$, for the first time, in this same order.
By Theorem \ref{Thm:OrderSpecialCrossings}, the respective intervals of the above double crossings form a nested sequence $I_1\supset I_2\supset \ldots \supset I_k$.

\begin{figure}[htbp]
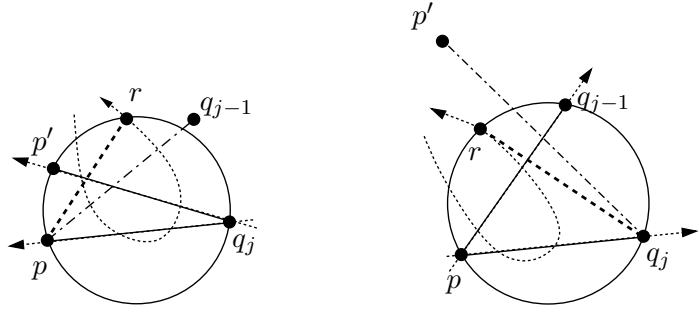

\begin{center}
\input{QuadraticDouble1.pstex_t}\hspace{2cm}\input{QuadraticDouble2.pstex_t}
\caption{\small Proof of Theorem \ref{Thm:SpecialCrossings}. Left: If the double crossing $(p'q_j,r,I')$ ends before the end of $I_{j-1}$ then the second co-circularity of $q_j,p,p',r$ occurs during $I_{j-1}$. Right: If the double crossing $(p'q_j,r,I')$ ends after $I_{j-1}$ then the second co-circularity of $p,q_{j-1},q_j,r$ occurs during $I'$.}
\label{Fig:QuadDouble}
\end{center}
\end{figure} 

Clearly, the first crossing $(pq_1,r,I_1)$ can be uniquely charged to the pair $p,r$.
Now assume that $k>1$. We show that each of the additional double Delaunay crossings $(pq_j,r,I_j)$, for $2\leq j\leq k$, 
can be uniquely charged to the corresponding pair $q_j,r$.
Specifically, we show that no double Delaunay crossing of incoming $q_j$-edges $p'q_j$ (that is, $p$-edges that are oriented towards $p$), by $r$, can end after $I_j$. In other words, $(pq_j,r,I_j)$ is the ``last" such double crossing.

Indeed, fix $2\leq j\leq k$ as above. We first show that no double crossing of the form $(p'q_j,r,I')$ can end during the interval which lasts from the end of $I_j$ and to the end of $I_{j-1}$. Indeed, suppose to the contrary that such a situation occurs, and apply a suitable variant of Theorem \ref{Thm:OrderSpecialCrossings} to the double Delaunay crossings of $q_j$-edges $p'q_j$ and $pq_j$ by $r$. 
By Part (iv) of that theorem, $I_j$ is contained in $I'$, and the four points $q_j,p,p',r$ are involved in a red-blue co-circularity with respect to $p'q_j$ during the second portion of $I'\setminus I_j$. See Figure \ref{Fig:QuadDouble} (left). Right after that co-circularity, the Delaunayhood of $pr$ is violated by $q_j$ and $p'$. If $I'$ ends before the end of $I_{j-1}$, the above co-circularity must occur during $I_{j-1}$ (as $I_{j-1}\supset I_j$), which contradicts Lemma \ref{Lemma:Crossing} (applied to the crossing of $pq_{j-1}$ by $r$).

It remains to show that no double Delaunay crossing $(p'q_j,r,I')$, as above, can end after the end of $I_{j-1}$. Indeed, by Part (iv) of Theorem \ref{Thm:OrderSpecialCrossings} (now applied to the double crossings of the $p$-edges $pq_{j-1}$ and of $pq_j$, by $r$), the points $p,q_{j-1},q_j,r$ are involved in a co-circularity during the second portion of $I_{j-1}\setminus I_j $. Right after this co-circularity, the Delaunayhood of $q_jr$ is violated by $p$ and $q_{j-1}$. 
If the interval $I'$ (which contains $I_j$) ends after the end of $I_{j-1}$, the aforementioned co-circularity must occur during $I'$; see Figure \ref{Fig:QuadDouble} (right). However, this is another contradiction to Lemma \ref{Lemma:Crossing} (now applied to the crossing of $p'q_j$ by $r$, which takes place during $I'$).

We have shown that every double Delaunay crossing can be uniquely charged to an (ordered) pair of points of $P$, so their number is $O(n^2)$, as asserted.
\end{proof}

\section{Conclusion} We have studied the number of discrete changes in the Delaunay triangulation of a set $P$ of $n$ points moving along pseudo-algebraic trajectories in the plane, so that any four points of $P$ can be co-circular at most {\it twice} during the motion. 
We have introduced a new concept of Delaunay crossings, and established several interesting structural properties of these crossings.
In our analysis we have used Theorem \ref{Thm:RedBlue} to reduce the problem of bounding the number of Delaunay co-circularities to the more specific problem of bounding the number of Delaunay crossings. Notice that the proof of Theorem \ref{Thm:RedBlue} did not rely on any assumptions concerning the motion of the points of $P$ (except for its being pseudo-algebraic of constant degree). Moreover, the aforementioned reduction easily extends to 
the case in which any four points of $P$ can be co-circular at most {\it three times} during the motion.
For these reasons, the author believes that the techniques introduced in this paper can be used to establish sub-cubic upper bounds for more general instances of the problem, such as the instance where the points are moving along straight lines with equal speeds.

\section{Acknowledgements} I would like to thank my former Ph.D. advisor Micha Sharir 
whose help made this work possible.
In particular, I would like to thank him for the insightful discussions, and, especially, for his invaluable help in the preparation of this paper. 

In addition, I would like to thank the anonymous
DCG referees for valuable suggestions that helped to improve the
presentation.



\end{document}